\documentclass[journal,10pt]{IEEEtran}
  \usepackage{cite}
\IEEEoverridecommandlockouts
\usepackage{paralist}
\usepackage{amsmath,amssymb,amsfonts}
\usepackage{graphicx}
\usepackage{paralist}
\usepackage{textcomp}
\usepackage{mathtools}
\usepackage{xcolor}
\usepackage{color}
\usepackage{algorithmic}
\usepackage{algorithm}
\usepackage{tcolorbox}
\usepackage{multirow}
\usepackage{lipsum}
\usepackage{verbatim}
\usepackage[normalem]{ulem}
\usepackage{dsfont}
\usepackage{listings}
\usepackage{url}
\usepackage[utf8]{inputenc}
\graphicspath{ {./images/} }
\usepackage{acronym}
\usepackage{cite}
\usepackage{lettrine}
\usepackage{subfigure}
\usepackage{enumitem}
\usepackage{array}
\usepackage{amsthm}

\newtheorem{theorem}{Theorem}[]

\newtheorem{lemma}{Lemma}[]

\def\BibTeX{{\rm B\kern-.05em{\sc i\kern-.025em b}\kern-.08em
    T\kern-.1667em\lower.7ex\hbox{E}\kern-.125emX}}
    \newcolumntype{M}[1]{>{\centering\arraybackslash}m{#1}}
\begin{document}
\newpage
\title{\fontsize{22.8}{27.6}\selectfont Handover Management in UAV Networks with Blockages}

% Impact of Device Caching and Handovers on the Performance of 3D UAV Networks with Blockages}

\author{Neetu R.R, Gourab Ghatak and Vivek Ashok Bohara\thanks{Neetu R.R. and Vivek Ashok Bohara are affiliated with the Department of Electronics and Communication Engineering, IIIT-Delhi, India. (Email: {neetur, vivek.b}@iiitd.ac.in). Gourab Ghatak is with the Department of Electrical Engineering, IIT Delhi, India. (Email: gghatak@ee.iitd.ac.in). This research is funded by the IIT Palakkad Technology IHub Foundation Doctoral Fellowship IPTIF/HRD/DF/026.
}}
% \thanks{Neetu R.R and Vivek Ashok Bohara are with the Department of Electronics and Communication Engineering, IIIT-Delhi, India. (email: \{neetur, vivek.b\} @iiitd.ac.in) and Gourab Ghatak is with the Department of Electrical Engineering, IIT Delhi, India (email: gghatak@ee.iitd.ac.in).}

\acrodef{PPP}[PPP]{Poisson point process} 
\acrodef{BPP}[BPP]{Binomial point process}
\acrodef{HPPP}[HPPP]{homogeneous Poisson point process}
\acrodef{SBS}[SBS]{small base station} 
\acrodef{CSP}[CSP]{conditional success probability} 
\acrodef{A2G}[A2G]{air-to-ground}
\acrodef{TBS}[TBS]{terrestrial base station}
\acrodef{BS}[BS]{base station}
\acrodef{MBS}[MBS]{macro base station}
\acrodef{CDF}[CDF]{cumulative distribution function}
\acrodef{CCDF}[CCDF]{complementary-cdf}
\acrodef{SINR}[SINR]{signal-to-interference noise ratio}
\acrodef{MCP}[MCP]{Matérn cluster process}
\acrodef{QoS}[QoS]{quality of service}
\acrodef{QoE}[QoE]{quality of experience}
\acrodef{UDN}[UDN]{ultra-dense network}
\acrodef{UAV}[UAV]{unmanned-aerial vehicle}
\acrodef{ASE}[ASE]{area spectral efficiency}
\acrodef{URLLC}[URLLC]{ultra-reliable low-latency communications}
\acrodef{MPPP}[MPPP]{marked-Poisson point process}
\acrodef{PCP}[PCP]{poisson cluster process}
\acrodef{HO}[HO]{handover}
\acrodef{UE}[UE]{user equipment}
\acrodef{MD}[MD]{meta distribution}
\acrodef{RSRP}[RSRP]{reference signal received power}
\acrodef{PGFL}[PGFL]{probability generating functional}
\acrodef{RSS}[RSS]{received signal strength}
\acrodef{PDF}[PDF]{probability density function}
\acrodef{CSI}[CSI]{channel state information}
\acrodef{VT}[VT]{vehicle terminal}
\acrodef{MT}[MT]{mobile terminal}
\acrodef{5G}[5G]{fifth-generation}
\acrodef{3D}[3D]{3-dimension} 
\acrodef{UAV}[UAV]{unmanned aerial vehicle}
\acrodef{AV}[AV]{aerial vehicle}
\acrodef{UAV-BS}[UAV-BS]{UAV-base station}
\acrodef{UAV-AP}[UAV-AP]{UAV-access point}
\acrodef{LoS}[LoS]{line-of-sight}
\acrodef{MD}[MD]{meta distribution}
\acrodef{UAV}[UAV]{unmanned aerial vehicle}
\acrodef{CCDF}[CCDF]{complementary cumulative distribution
function}
\acrodef{RSSI}[RSSI]{received signal strength indicator}
\acrodef{NLoS}[NLoS]{non-line-of-sight}
\acrodef{HetNet}[HetNet]{heterogeneous network}
\acrodef{IoT}[IoT]{Internet of Things}
\acrodef{3-D}[3-D]{three-dimensional}
\acrodef{2-D}[2-D]{two-dimensional}
\acrodef{1-D}[1-D]{one-dimensional}
\acrodef{EC}[EC]{edge computing}
\acrodef{MLD}[MLD]{mean local delay}
\maketitle

\begin{abstract}
We investigate the performance of \ac{UAV}-based networks in urban environments characterized by blockages, focusing on their capability to support the service demands of mobile users. The \acp{UAV-BS} are modeled using a \ac{2-D} \ac{MPPP}, where the marks represent the altitude of each \ac{UAV-BS}. Leveraging stochastic geometry, we analyze the impact of blockages on network reliability by studying the \ac{MD} of the \ac{SINR} for a specific reliability threshold and the association probabilities for both \ac{LoS} and \ac{NLoS} \acp{UAV-BS}. Furthermore, to enhance the performance of mobile users, we propose a novel cache-based handover management strategy that dynamically selects the cell search time and delays the \ac{RSS}-based \ac{BS} associations. This strategy aims to minimize unnecessary \acp{HO} experienced by users by leveraging caching capabilities at \ac{UE}, thus reducing latency, ensuring seamless connectivity, and maintaining the \ac{QoS}. This study provides valuable insights into optimizing UAV network deployments to support the stringent requirements in the network, ensuring reliable, low-latency, and high-throughput communication for next-generation smart cities.
\end{abstract}


\section{Introduction}
\lettrine{U}{nmanned} aerial vehicles equipped with remote radio heads (RRHs) function as relays and aerial base stations~\cite{29}. With continuous advancements in drone technology, employing UAVs as airborne \acp{BS} offers a cost-effective and scalable solution to enhance coverage and improve \ac{QoS} for end-users~\cite{12} -\nocite{13} \cite{14}. They can provide on-demand connectivity and extend network coverage to underserved or temporarily congested urban areas. Urban environments, however, present unique challenges for UAV networks, such as blockages, user mobility, and the increasing demand for reliable, low-latency communications~\cite{20}. With the rapid growth of smart cities and applications requiring real-time data processing, such as autonomous driving and mobile gaming, ensuring uninterrupted service is critical~\cite{16}. Addressing these challenges is vital, as they significantly impact the development and effectiveness of next-generation wireless networks. This research focuses on enhancing UAV network reliability and performance in urban settings by improving connectivity in poorly covered areas, meeting growing mobile data demands, and utilizing device caching for better service continuity and efficiency~\cite{15}. These efforts lay the foundation for scalable and adaptive UAV deployments in smart cities. Additionally, performing a thorough blockage analysis during the deployment planning of \acp{UAV-BS} is essential to maintain stable and reliable communication links with ground users. The impact of blockages must also be a primary consideration when developing association strategies in heterogeneous networks~\cite{19}. Despite their potential, studying \ac{UAV} networks in urban environments is inherently complex due to the stochastic nature of blockages, the dynamics of user mobility, and the need to maintain consistent \ac{QoS}~\cite{18}. Integrating device caching capabilities and optimizing UAV placement and handover strategies further adds to this complexity.

Consequently, in
the proposed work, we address the challenges posed by urban blockages in UAV-based cellular networks, which severely impact \ac{LoS} connections essential for smart city applications. By leveraging stochastic geometry, we model these blockages and analyze network reliability through the \ac{MD} of \ac{SINR}, offering a comprehensive evaluation of \ac{QoS}. To address frequent handovers caused by user mobility, we propose a cache-enabled \ac{HO} management scheme that utilizes caching at the \ac{UE} to reduce unnecessary handovers, lower latency, and enhance throughput. The proposed cache-enabled handover management strategy provides a practical approach to ensuring seamless connectivity, meeting stringent QoS requirements, and dynamically adapting to user mobility and blockage scenarios. 

\subsection{Related Works}
Drones as relays and aerial \acp{BS} gain interest due to their flexibility in deployment and flying capabilities\cite{21}. 
Deploying \acp{UAV} in regions lacking \acp{TBS} to enhance capacity and coverage has been broadly investigated in the literature e.g., \cite{22} -\nocite{23} \cite{24}. A number of studies have so far investigated incorporating \acp{UAV} into the network, and many of them have adopted stochastic geometry tools to analyze the network performance~\cite{7}, \cite{25}. The authors in \cite{26} investigated the feasibility of integrating two applications of \acp{UAV}, such as aerial \acp{BS} and data delivery, serving multiple IoT clusters on the ground. They proposed a performance metric, data delivery efficiency, to deliver the data to the clusters in less time. They proposed an algorithm to jointly optimize the minimum round trip time to serve the clusters and maximize the delivered data to the clusters. 

\subsubsection{Blockage Analysis}
 The authors in \cite{1} used the concepts of random shape theory to model the blockages with random sizes, shapes, and orientations. The blockages are considered as \ac{PPP}, and they incorporated the heights of the blockages for the network analysis. They study the behavior of the probability of the \ac{LoS} links and its dependency on the length of the link. They analyzed the coverage probability and achievable rate at the user in the presence of blockages.
In \cite{5}, the authors proposed a power control strategy for urban \ac{UAV} networks that mitigates interference by muting the transmissions of interfering UAVs. The proposed strategy includes a power control coefficient that is dependent on environmental blockage parameters and uses stochastic geometry to derive coverage probability and network connectivity. The authors in \cite{36} study multiple types of blockages in the network, which affect the coverage performance of single-swarm mmWave \ac{UAV} networks. The effect of blockages, such as static, dynamic, and self-blockages, on the system performance are jointly considered. They analyzed the effect of all these blockages on the coverage probability, average \ac{LoS} probability, and average path loss for UAVs by varying the number of \acp{UAV} and heights of the \acp{UAV}. They observed that static blockage and self-blockage exert a predominant influence on these metrics, and the effect of self-blockage can be mitigated by flying the \acp{UAV} at a sufficiently high altitude.

\subsubsection{Fine-Grained Analysis of UAV Networks}
Meta distribution and reliability analysis in traditional cellular networks is studied in various literature~\cite{41}\nocite{42}-\cite{43}. The fine-grained analysis in UAV-enabled networks is performed in \cite{10}, which investigates the \ac{SINR} \ac{MD} performance of \acp{UAV} equipped with either steerable or fixed directional antennas using stochastic geometry. They derived the distribution of the off-boresight angle (OBA) and calculated the moments of the \ac{CSP}. Additionally, they obtained the mean local delay and SINR MD for various environmental scenarios. The study also explored the asymptotic behavior of the association probability and the moments of the \ac{CSP}. Moreover, the authors in \cite{38} analyzed the downlink performance of the UAV network from the perspective of each link other than spatial averaging. The authors consider two types of user distributions, \ac{PPP} and \ac{MCP}, and observed the \ac{MD} for UAV-enabled HetNet and \ac{TBS} only network. The authors the advantages and disadvantages of deploying \acp{UAV} and the height and density at which the \acp{UAV} should be deployed in the network depending on different environmental scenarios. The authors in \cite{39} study the downlink rate meta-distribution in a typical UAV in a cellular-connected UAV network. Considering the Nakagami-m fading for the \ac{LoS} transmissions, the authors derive \ac{LoS} probability and observe the impact of rate threshold and link reliability threshold on rate \ac{MD}. 
\subsubsection{Handover Analysis}
\ac{HO} management using different schemes, e.g., velocity-aware HO management~\cite{46}, route-aware HO management~\cite{45} and skipping alternatively~\cite{43} and other schemes have been discussed in the literature \cite{16} -\nocite{43}\nocite{29}\nocite{45}\nocite{47}\cite{48}. Different from the above literature, the authors in~\cite{27} proposed a periodic handover skipping strategy that introduces a defined skipping period after each handover. They used stochastic geometry to analyze the handover rate and the expected download rate at the user end. In addition, \cite{27} also defined a utility metric, which is a function of the HO rate and achievable download rate, and compared the utility function for different distributions of user velocity. It is observed that when the user's velocity is very high, the proposed strategy outperforms conventional methods. 
% In \cite{16}, the authors examine the benefits of caching as a remedy for critical handover issues. Their focus includes handover failures, handover rates, and the load on target \ac{BS}. They utilize millimeter wave (mmWave) connections on a BS operating in dual mode, which allows mobile devices to cache requested data, ultimately reducing the number of handovers experienced by users. They estimated the average rate of caching at the \ac{UE} and examined the impact of user velocity on \ac{HO} failure. Additionally, an optimization problem is formulated aiming to analyze load balancing between the \ac{MBS} and \ac{SBS}, thus enhancing traffic offload from the \ac{MBS}. 
In \cite{28}, the authors investigated a \ac{3-D} mobile UAV network and analyzed the handover and coverage probability using stochastic geometry tools. They discussed two types of association strategies of the typical user in the origin, such as strongest average \ac{RSS} association and nearest distance association. The authors analyzed the handover probability and coverage probability of these two association strategies. They observed that the \ac{RSS} based strategy performs better than the distance-based associations, and there exists an optimal UAV density and altitude to maximize the coverage probability. In our work in \cite{9}, we study a handover management scheme by leveraging device caching in vehicles where the locations of \acp{SBS} are distributed according to a \ac{1-D} \ac{HPPP}. We analyze a handover-skipping strategy by storing the necessary data at the vehicle terminal to reduce unnecessary handovers while maintaining the \ac{QoE} at the user end. We derive the analytical expressions for caching distance, \ac{HO} rate, and the average throughput experienced by the vehicle terminal. We perform a spatiotemporal expectation to derive the effective average rate experienced by the vehicle terminal.

\subsection{Motivation and Contribution}
The deployment of \acp{UAV-BS} has emerged as a promising solution for enhancing network performance in dynamic, blockage-prone urban environments. However, most prior studies assume \acp{UAV} are deployed at fixed altitudes, limiting their mobility to a single axis. In this paper, we overcome this limitation by modeling \acp{UAV-BS} as a 2-D \ac{MPPP}, representing their locations in \ac{3-D} space with random distributions over 
$\mathbb{R}^2$ in the horizontal plane and along the vertical axis. Different from our conference version in \cite{49}, this study extends the analysis to urban scenarios characterized by blockages and mobile users. Notably, this study explores the \ac{MD} of SINR for a specified reliability threshold in networks where UAVs are distributed as a \ac{2-D} \ac{MPPP}, considering the impact of blockages, is an aspect which is not previously addressed in the literature. Moreover, none of the literature has proposed a novel HO management scheme that leverages the caching capabilities of \acp{UE}. This approach aims to mitigate handover delays and enhance \ac{QoS} in a 2-D \ac{MPPP} UAV network under blockage conditions.

 % None of the aforementioned studies examined a dynamic environment or the effect of user mobility on network performance, including the \ac{HO} rate and the average throughput experienced by users as a result of handovers.

% None of the above literature examines the impact of blockages on the coverage probability, \ac{HO} rates, and network reliability. Additionally, the impact of user mobility on the \ac{QoS} in terms of average throughput and reliability remains unexplored. In most of the prior works, the altitude of UAVs is either fixed or at different tiers. A height-dependent deployment of \acp{UAV-BS} in dynamic scenarios in the presence of blockages is necessary to maintain LoS connectivity and optimize network performance.

Motivated by this, we investigate a \ac{2-D} \ac{MPPP} in which the altitudes of the \acp{UAV-BS} are represented as marks. By leveraging stochastic geometry, this work analyzes \ac{SINR} \ac{MD} for a specific reliability threshold in the presence of blockages. This work proposes a novel cache-enabled HO management scheme, leveraging caching at \ac{UE} to minimize unnecessary \acp{HO}, reduce delays, and ensure seamless service continuity.
% \subsection{Contributions and Organization}

The significant contributions are outlined as follows:
\begin{enumerate}
    \item We develop a novel framework to represent UAV-BS locations in \ac{3-D} space, where the \ac{2-D} positions follow a \ac{MPPP} and altitudes are modeled as random marks in an urban network with blockages. We derive analytical expressions for the distance distributions to the nearest \ac{LoS} and \ac{NLoS} UAV-BSs from the typical user, as well as for the probabilities of user association.
    % This approach enables UAV deployments to align with user requirements, line of sight restrictions, and mobility, addressing the limitations of conventional strategies based on simplistic horizontal or vertical movements.    
    % We obtain the distribution of the distance from the typical user to the nearest \ac{LoS} and \ac{NLoS} \acp{UAV-BS}, modeled as a 2-D \ac{MPPP} with altitude serving as the marks while accounting for blockages. Additionally, we derive the probability that the typical user associates with \ac{LoS} and \ac{NLoS} UAV-BSs in a 2-D \ac{MPPP}. 
    \item 
    Given the proposed network topology, we analyze the \ac{MD} of \ac{SINR} for a specified reliability threshold in UAV-based urban networks. Our study investigates the impact of varying blockage densities and UAV deployment densities, determining the optimal deployment density that enhances overall network reliability while minimizing the mean local delay. 

    \item Additionally, we propose a cache-enabled \ac{HO} management algorithm that leverages caching capabilities at \ac{UE} to minimize unnecessary \acp{HO}, reduce HO delays, improves \ac{QoS} and energy efficiency. This strategy dynamically adjusts cell search intervals and employs an HO skipping mechanism to enhance the user experience in dense UAV urban networks with blockages. 
    \item \textit{Insights:}
    \begin{itemize}
        \item We observe that increasing network density does not consistently enhance \ac{LoS} associations. On the contrary, higher densities expand the NLoS regions relative to \ac{LoS} regions, thereby reducing the number of \ac{LoS} connections from the typical user.    \item Contrary to traditional assumptions, for higher deployement density of UAV-BSs, blockages can enhance the network performance by blocking the interference in the network. In regions with lower blockage levels, fewer UAVs are demanded, while regions with higher blockage levels demand a larger number of UAVs to achieve improved network performance.
        \item The intensity and distribution of blockages play a crucial role in determining the optimal deployment altitude for \acp{UAV-BS}. In regions with sparse blockages, lower altitudes improve performance by maximizing \ac{LoS} connections, whereas in densely blocked areas, higher altitudes are preferable to mitigate obstructions and enhance connectivity. 

    \end{itemize}
    \end{enumerate}
The rest of the paper is organized as follows. In Section~\ref{sec:SM}, we discuss the system model and summarize the objectives. In Section~\ref{DD_MD}, we discuss the appropriate distance distributions to the nearest UAV-BSs and association probabilities. In Section~\ref{CSP_MD}, \ac{MD}-based analytical framework and \ac{CSP} analysis is performed. In Section~\ref{HO_AT}, a cache-based handover skipping algorithm is proposed and derived \ac{HO} rate and average throughput. Section~\ref{NM_D} validates the analytical framework and presents numerical results that highlight key aspects of the network. Finally, the conclusions are discussed in Section~\ref{CL}.

\section{System Model}
\label{sec:SM}

% Table~\ref{TABLE: table1} provides the notations used in the paper.

% \renewcommand{\arraystretch}{0.6}
% \begin{table*}[h!]
% \centering
%      \caption{Notation}
%     \label{TABLE: table1}
%     \begin{tabular}{|M{1cm}|M{6cm}|M{1cm}|M{6cm}|}\hline
%     \small
%     Symbol & Definition & Symbol & Definition\\ \hline
%     $\lambda_u$ & \ac{SBS} intensity & $f_c$ & Carrier frequency \\\hline
%     $\lambda_b$ & Blockage intensity & $\eta$ & Probability of blocking 3-D link\\\hline
%     $H_b$ & Average Blockage height & $q,p$ & Parameters of LoS probability\\\hline
%     $h_{min}$ & Minimum UAV height & $h_{max}$ & Maximum UAV height\\\hline
%     $B$ & Bandwidth & $\sigma_N$ & Noise power\\\hline
%     $P_u$ & Transmit power & $\alpha_L$ & LoS path-loss exponent\\\hline
%     $v$ & User velocity & $s_r$ & Service rate requirement \\\hline
%     $K$ & Path loss coefficient & $\alpha_N$ & NLoS path-loss exponent\\\hline
%     $t_s$ & Search time & $T$ & Total time frame\\\hline
%     $R_{\mathrm{eff}}$ & Effective average throughput & $R_{\rm{ca}}$ & Average throughput  \\\hline
%     $\mu$ & Average handover rate & $G$ & Cache size\\\hline
%     $w_{\rm{c}}$  & Power efficiency of caching & $t_H$ & \ac{HO} overhead time \\\hline
%     $N$  & Number of blockages & $G_L$ & Nakagami-m fading\\\hline
%     $g_N$  & Rayleigh fading & $d_L$ & LoS UAV-user distance\\\hline
%     $d_N$  & NLoS UAV-user distance & $L_S$ & LoS UAV-probability\\\hline
%      $N_S$ & NLoS UAV-probability & $\bar{N}$ & Number of blockages obstructing 3-D link\\\hline
     
%     \end{tabular}
%   \end{table*}

\subsection{Network Model}
We consider a downlink \ac{UAV}-based network in urban environments characterized by blockages, designed to support the connectivity demands of smart device users. The network consists of \acp{UAV-BS} modeled as a \ac{2-D} homogeneous \ac{PPP}, $\Phi_U$ on $\mathbb{R}^2$, with intensity $\lambda_u$. Each point of $\Phi_U$ has an associated mark representing the altitude of the \acp{UAV-BS}. We assume the marks are uniformly distributed between $h_{min}$ and $h_{max}$. This configuration forms a \ac{2-D} \ac{MPPP}. We assume that the \acp{UAV-BS} are equipped with down-tilted omnidirectional antennas serving single-antenna users. Without loss of generality, utilizing the stationarity and isotropy properties of the \ac{PPP}~\cite{2}, we consider a user moving in a straight line within a \ac{2-D} plane that passes through the origin at a velocity 
$v$. Let $\textbf{X}_0 \in \Phi_U$ represent the 2-D location of the \ac{UAV-BS} at an altitude $h_0$ to which the user is connected at a specific time. Without loss of generality, leveraging the ergodicity property of the \ac{PPP}~\cite{2}, we assume the user to be located at the
origin at that time. By averaging over various realizations, this user is treated as the typical user. 
Using the Uu interface~\cite{50}, the typical user associates with the \ac{UAV-BS} providing the maximum received power, assessed through the estimated \ac{RSSI} measurements in the downlink. Each \ac{UAV-BS} has a dedicated backhaul link to the centralized network. Additionally, we consider that orthogonal frequency allocation is employed to aggregate users~\cite{4}. Since the altitudes of UAV-BSs are uniformly distributed between $h_{min}$ to $h_{max}$, the cell boundaries are formed by weighted Poisson Voronoi tessellation in the 2-D plane. Another tier of \acp{TBS} can easily be incorporated into the network without any change in the mathematical analysis.
% The typical user performs \ac{HO} at every cell boundary, leading to a degradation in user performance. 

\begin{figure}[htbp]
   \begin{minipage}[t]{0.45\textwidth}
    \centering\includegraphics[width=1\linewidth,height=0.6\linewidth]{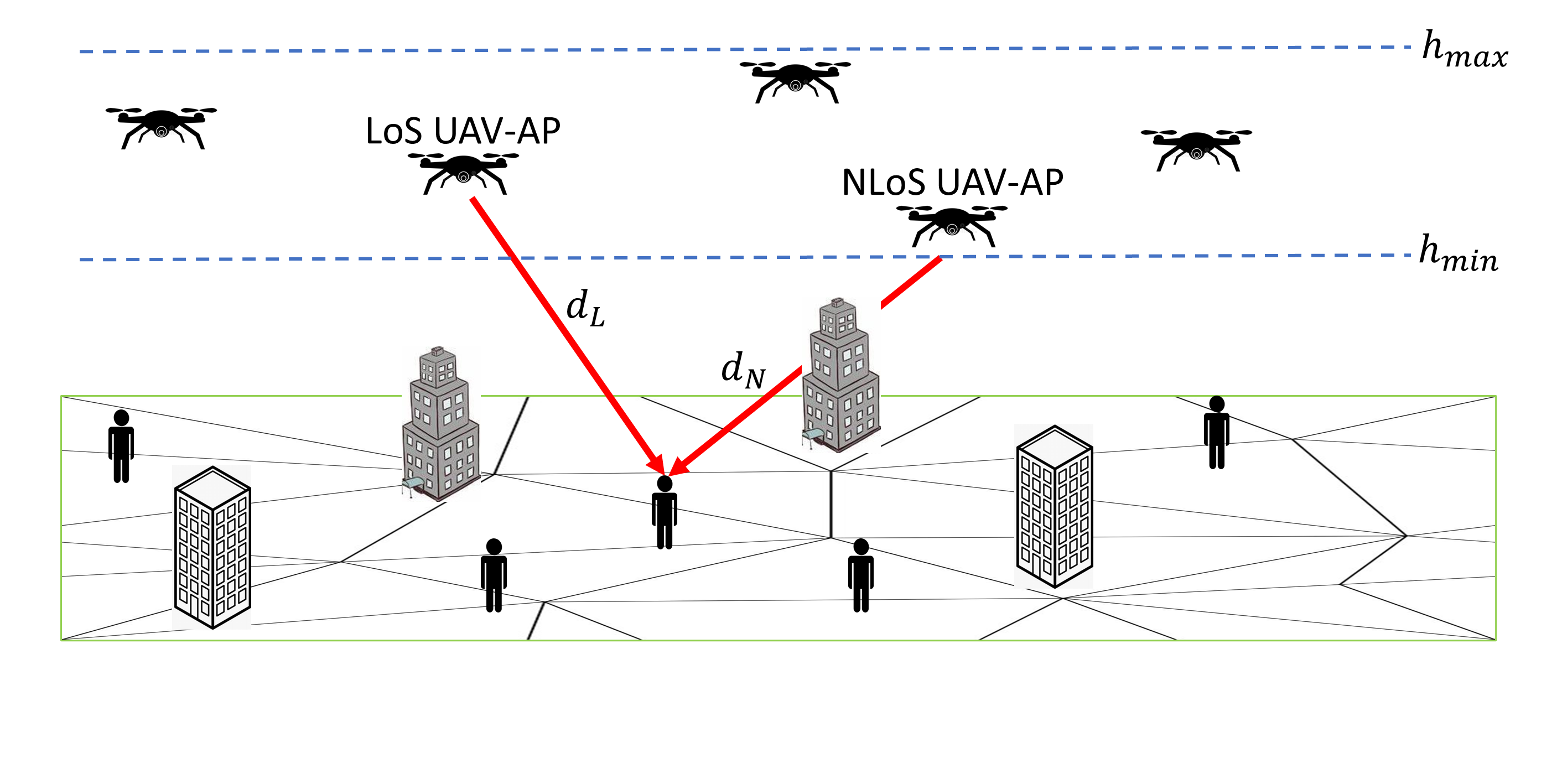}
    \caption{An illustration of the system model in which the user is served by either \ac{LoS} or \ac{NLoS} UAV-BS.}
    \label{fig:system_main}
    \end{minipage}\hfill
    \begin{minipage}[t]{0.45\textwidth}
     \centering
    \includegraphics[width=0.7\linewidth,height=0.5\linewidth]{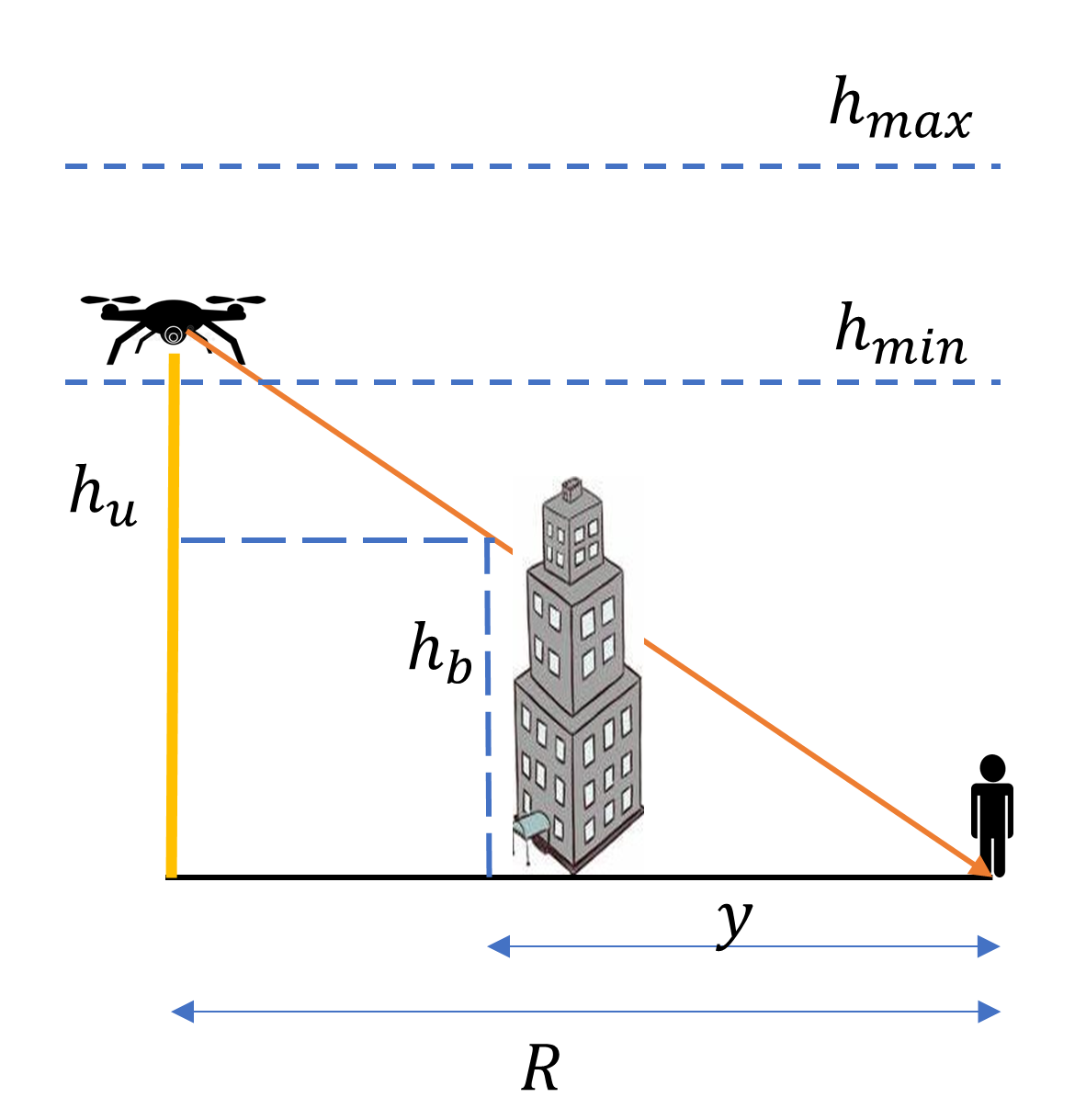}
    \caption{A building intersecting the 2-D link of length $R$ between the user and UAV-BS intersects the 3-D link between user and UAV-BS at the altitude of $h_u$ if and only if this altitude is greater than $h_b$ as in the figure.}
    \label{fig:block}
    \end{minipage}
    \end{figure}

% \begin{figure}
% \centering
% \includegraphics[width=.5\linewidth,height=.25\linewidth]{model.png}
% \vspace{-1cm}
% \caption{An illustration of the system model in which the user is served by either \ac{LoS} or \ac{NLoS} UAV-BS.}
% \label{fig:system_main}
% \end{figure}
\subsection{Blockages Process}
The randomly positioned buildings are modeled using a Boolean scheme of random rectangles~\cite{1}. The blockage centers represented as ${C_k}$ is a homogeneous \ac{PPP} $\Phi_b$ of intensity $\lambda_b$. By the definition of \ac{PPP}, the blockages are placed independently of one another, ensuring the geometry of the blockages does not intersect. The lengths and widths of the blockages are assumed to be independent and identically distributed (i.i.d.), following the probability distributions $f_L(l)$ and $f_W(w)$, respectively. The intensity, length, and width of the blockages are selected in such a way that the typical user moving in a straight line does not pass through the blockages. The probability of intersection between the user and the blockage is almost zero. Additionally, the orientation of the blockages, $\theta_k$ is uniformly distributed over the interval (0,$2\pi$].  The heights of the blockages are modeled by a Rayleigh distribution with \ac{PDF} $f_H(h)$~\cite{5}. We observe the number of blockages $N$ crossing the 2-D link between the user and the UAV-BS, which is a Poisson-distributed random variable. 
 From Fig.~\ref{fig:block}, $R$  represents the 2-D distance between the user and the \ac{UAV-BS}. $y$ indicates the distance from the user to the point where a building intersects the 2-D link of length $R$. Given that $N$ buildings intersect the link, their intersection points are uniformly distributed along the interval [0, $R$]~\cite{1}. Thus, $y$ is a random variable that follows a uniform distribution from [0,$R$]~\cite{35}.

%  \begin{figure}
% \centering
% \includegraphics[width=.3\linewidth,height=.25\linewidth]{block.png}
% \caption{A building intersecting the 2-D link of length $R$ between the user and UAV-BS intersects the 3-D link between user and UAV-BS at the altitude of $h_u$ if and only if this altitude is greater than $h_b$ as in the figure}
% \label{fig:block}
% \end{figure}

 % \textbf{Remark:} While the building centers are placed according to a Poisson point process in 2D space, when restrict our focus to the buildings intersecting the 1D link, their positions on that link become uniformly distributed.

\begin{lemma}
Given that a blockage intersects the 2-D link between the user and the UAV-BS of distance $R$ at a height $h_u$, the conditional probability that it blocks the 3-D Euclidean link between the user and the UAV-BS is given as
\begin{equation}
    \mathbb{P}(h>h_b|y,h_u)=1- \int_0^{\frac{y h_u}{R}} f_H(h) dh,
    \label{eq1}
\end{equation}   where $h_b=\frac{y h_u}{R} $.  Marginalizing over $y$ and $h_u$,
% \begin{equation}
%     \eta=\frac{1}{(h_{max}-h_{min})} \int_{h_{min}}^{h_{max}} \int_0^1 \exp{\Bigg(-\frac{h^2 s^2}{2 \sigma^2}\Bigg)} \mathrm{d}s \mathrm{d}h  ,   \label{eq3}
% \end{equation}
\begin{equation}
   \eta=\frac{1}{(h_{max}-h_{min})} \int_{h_{min}}^{h_{max}} \int_0^1  \exp{\Bigg(-\frac{h^2 s^2}{2 \sigma^2}\Bigg)}\hspace{0.1cm} \mathrm{d}s \mathrm{d}h  ,   \label{eq3} 
\end{equation}
if and only if $R>0$.
\end{lemma}
\begin{proof}
    As mentioned above, the heights of the blockages are Rayleigh distributed with \ac{PDF} given as $f_H(x)= \frac{x}{\sigma^2} \exp{-\frac{x^2}{2 \sigma^2}}.$ Leveraging the property of Rayleigh distribution, the average height of blockages $H_b= 1.253\sigma$.

The probability that the 3-D link is blocked is given as
\begin{equation}
    \mathbb{P}(h>h_b|y,h_u)=1- \int_0^{\frac{y h_u}{R}} f_H(h) dh= \exp{\Bigg(-\frac{h_u^2 \Big(\frac{y}{R}\Big)^2}{2 \sigma^2}\Bigg)}.
    \label{eq4}
\end{equation}
Marginalizing (\ref{eq4}) with respect $y$,
\begin{equation}
    \mathbb{P}(h>h_b|h_u)= \frac{1}{R} \int_0^R  \exp{\Bigg(-\frac{h_u^2 \Big(\frac{t}{R}\Big)^2}{2 \sigma^2}\Bigg)}\hspace{0.1cm} \mathrm{d}t.  
    \label{eqx}
\end{equation}
Substituting $\frac{t}{R}=s$, and taking an expectation over $h_u$, uniformly distributed between $h_{min}$ and $h_{max}$, we obtain (\ref{eq3}).
% \begin{equation}
%     \mathbb{P}(h>h_b|h_u)= \int_0^1  \exp{\Bigg(-\frac{h_u^2 s^2}{2 \sigma^2}\Bigg)}\hspace{0.1cm} \mathrm{d}s.  
%     \label{eqxs}
% \end{equation}

% Taking an expectation over $h_u$, which is uniformly distributed between $h_{min}$ and $h_{max}$, we obtain (\ref{eq3}).
\end{proof}
\textbf{Remark:} When 
$R=0$, the UAV-BS is directly above the user. Substituting $R=0$ in (\ref{eq4}), then the conditional probability $\mathbb{P}(h>h_b|y,h_u)=0$ holds. If we marginalize over 
$y$ and $h_u$, the above un-conditioned probability $\eta=0$. 
The total number of blockages crossing the 2-D link between the user and the UAV-BS is $N$. As the blockages are distributed as a \ac{PPP}, $N$ is a poison random variable, and its mean is $q R + p$, where $R$ is the length of the 2-D link~\cite{1}. 
\begin{equation}
    \mathbb{E}[N]=q R + p
    \label{eq6}
\end{equation}
where $q=\frac{2 \lambda_b (\mathbb{E}[L_b]+\mathbb{E}[W_b])}{\pi}$ and $p= \lambda_b \mathbb{E}[L_b] \mathbb{E}[W_b]$,
where $\mathbb{E}[L_b]$ and $\mathbb{E}[W_b]$ are evaluated by the average length and width of the blockages. $\eta$ is independent on $N$ as the intersections form 
\ac{PPP} on the 2-D link of length $R$. Let $\bar{N}$ represent the number of blockages obstructing the 3-D link between the user and UAV-BS. By applying the independent thinning to $N$, the expected value of $\mathbb{E}[\bar{N}]= \eta\mathbb{E}[N]$.

% Due to the property of independent thinning~\cite{2}, the intensity of blockages blocking the 3-D Euclidean link between the UAV and the user is $\eta \lambda_b$.
% Therefore, expected number of blockages in a 3-D \ac{UAV-BS}-user link is given as \cite{1}
% \begin{equation}
%     \mathbb{E}[N]=q R + p
%     \label{eq6}
% \end{equation}
% where $q=\frac{2 \eta \lambda_b (\mathbb{E}[L]+\mathbb{E}[W])}{\pi}$ and $p=\eta \lambda_b \mathbb{E}[L] \mathbb{E}[W]$,
% where $\mathbb{E}[L]$ and $\mathbb{E}[W]$ are evaluated by the average length and width of the blockages.

\subsection{Propagation Model}
    The \acp{UAV-BS} can either be in \ac{LoS} or \ac{NLoS} state~\cite{7} depending on visibility from the typical user. Let the locations of the \ac{LoS} UAV-BSs be represented by $\Phi_L$ and the locations of \ac{NLoS} UAV-BSs be represented by $\Phi_N$. Therefore, $\Phi_U= \Phi_L \cup \Phi_N$. The \ac{LoS} transmissions experience Nakagami distributed fast-fading, $G_{L}$, with parameter $m$~\cite{40}. The \ac{NLoS} \ac{UAV-BS} transmissions experience Rayleigh-distributed fast-fading $g_N$, with the variance equal to 1.  The downlink transmit power of the UAV-BSs is $P_u$. For the large-scale path loss, we consider the received power at the typical user from the \ac{LoS} UAV-BS which is at a distance of $d_L$ is given by $R_L=P_u K G_L d_L^{-\alpha_L}$ and the \ac{NLoS} UAV-BS which is at a distance of $d_N$ given as $R_N=P_u K g_N d_N^{-\alpha_N}$. $\alpha_L$ is the pathloss exponent for LoS links and $\alpha_N$ is the pathloss exponent for NLoS links. The  probability of establishing an \ac{LoS} connection  between the \ac{UAV-BS} at a distance of $d$ and height $h$ from the typical user is given as
%     \begin{equation}
%     f_S= \left\{ \begin{matrix} \exp{(-q R+p)}; & S=1 (LoS)\\
%     1-\exp{(-q R+p}); & S=0 (NLoS)
% \end{matrix} \right . 
% \label{eq2}
% \end{equation}
\begin{equation}
    L_S(d,h)=\exp{\Big(-\eta(q \sqrt{d^2-h^2} +p)\Big)},
\end{equation}
where $h$ is the altitude of UAV-BS from the ground and $d$ is the 3-D Euclidean distance between the UAV-BS and the user, according to the visibility state of the UAV-BS. 
\begin{equation}
    d= \left\{ \begin{matrix} d_L; & \alpha=\alpha_L\\
    d_N; & \alpha=\alpha_N 
\end{matrix} \right . 
\end{equation}
Accordingly, the probability of NLoS transmissions is given
as $N_S(d,h)= 1- L_S(d,h)$. From (\ref{eq6}), we observe the probability of establishing \ac{LoS} link between the UAV-BS and the user is not only the function of the altitudes of UAV-BS and the Euclidean distance but also a function of the intensity of blockages, average length and width, and height of blockages.

\subsection{Performance Metrics}

The main factors driving 5G technology are reliability and latency, which form the foundation of \ac{URLLC} and support applications requiring real-time responsiveness and consistent performance. Considering this, we evaluate two critical performance metrics in a UAV network, where UAVs are modeled as a \ac{2-D} \ac{MPPP}, by analyzing the coverage probability or success probability experienced by the user. The first metric involves a fine-grained analysis of the \ac{SINR} experienced by the user, offering detailed insights into the reliability of communication links. This approach transcends traditional average \ac{SINR} evaluations by capturing link-level variability considering \ac{CSP}. The second metric addresses handover management, leveraging caching capabilities at the user terminal to reduce unnecessary \acp{HO} leading to latency and \ac{QoS} degradation during mobility. 
\begin{itemize}
\item \textbf{\ac{MD} of SINR:}
To perform a fine-grained analysis of \ac{SINR} in a stationary and ergodic point process setting, we calculate \ac{MD} of \ac{SINR}.  MD provides the distribution of the conditional success probability for a randomly selected user in a network. Unlike the mean SINR, which provides an average measure, the \ac{MD} of SINR provides a probabilistic framework to analyze the reliability of communication links in a wireless network. It represents the \ac{CCDF} of the success probability $\mathcal{P}_S(\gamma)$ conditioned on the point processes $\Phi_U$ and $\Phi_b$~\cite{6}.
\begin{equation}
   \bar{F}_{P_{S}}(\gamma, x_r)= \mathbb{P}^{!}(\mathcal{P}_S(\gamma) \hspace{0.1cm} > \hspace{0.1cm} x_r), \hspace{0.3cm} \gamma \in \mathbb{R}^+ ,\hspace{0.1cm}  x_r \in [0,1],
\end{equation}
where $\mathbb{P}^{!}$ represents the reduced Palm probability conditioning the typical user at the origin $o$, $x_r$ is the reliability threshold. The random variable $\mathcal{P}_S(\gamma)$ is given as
\begin{equation}
    \mathcal{P}_S(\gamma)= \mathbb{P}(\mathrm{SINR} > \gamma | \Phi_U, \Phi_b).
\end{equation}

\item{\textbf{Cache-enabled handover analysis:}} The rate requirement for the services for the users is very important for a continuous and uninterrupted \ac{QoS}. When the users are moving, it will lead to frequency handovers and handover failures, leading to degradation of this \ac{QoS}. In a traditional cellular setup, handovers happen at every cell boundary, leading to excessive delay in communication. In order to reduce the latency experienced by the user due to handovers, we explore a novel \ac{HO} management strategy aimed at minimizing handovers to provide a seamless user experience
across applications such as autonomous driving, navigation,
and video streaming.  According to 3GPP standards~\cite{33}, a user performs a cell search every $t_s$ seconds. A handover is triggered when a neighboring UAV-BS offers a \ac{RSS} higher than the associated UAV-BS. Each user requires
$\Delta T$ seconds to measure this \ac{RSS} from the neighboring cell. We dynamically select $t_s$ by leveraging caching capabilities in \ac{UE} in motion. We consider every user to have a rate requirement of $s_r$, and every \ac{UE} has a cache size of $G$ and moving with a velocity of $v$. 
% For a total travel time $T$, the analysis is performed at each coherence time $t_c$. 
Initially, cache memory is empty. When the downloaded data by the user exceeds the required service rate $s_r$, the difference is stored in the cache. Using this cached data, the user can delay cell searches and avoid unnecessary \acp{HO}, thereby reducing latency and maintaining a high \ac{QoS} during mobility. 
\end{itemize}

To obtain the above metrics, we derive the distance distribution of the associated UAV-BS from the typical user and the \ac{SINR} coverage probability experienced by the user.

\section{Characterization of Distance Distributions and Association Probabilities}
\label{DD_MD}
In this section, we obtain the distribution of the distance from a typical user to the nearest \ac{LoS} or \ac{NLoS} \ac{UAV-BS}, where the locations of \acp{UAV-BS} are distributed as a \ac{2-D} \ac{MPPP}. 
% The altitudes of the \acp{UAV-BS} are uniformly distributed between $h_{min}$ and $h_{max}$.

\begin{lemma}
Given that the typical user observes at least one \ac{LoS} and one \ac{NLoS} \ac{UAV-BS}, the \ac{CDF} for the distance distributions to the nearest LoS and NLoS UAV-BS distributed as a \ac{2-D} \ac{MPPP}, defined as $d_L$ and $d_N$ respectively, are given as
 \begin{equation}
    F_{d_{L}}(z)=\left \{\begin{matrix}
\bigg(1- \exp{\Big( A \times \int_{h_{min}} ^ z \mathcal{L}(z,h)\hspace{0.1cm} \mathrm{d} h \Big)}\bigg)/B_L; & \\ h_{min} \leq z < h_{max} \\ 
\bigg(1- \exp{\Big(A \times \int_{h_{min}} ^ {h_{max}} \mathcal{L}(z,h)\hspace{0.1cm} \mathrm{d} h  \Big)}\bigg)/B_L; & \\ z \geq h_{max}
    \end{matrix} \right.
    \label{cdf_z1}
\end{equation}

\begin{equation}
    F_{d_{N}}(z)=\left\{ \begin{matrix}
\bigg(1- \exp{\Big( A \times \int_{h_{min}} ^ z \mathcal{N}(z,h)\hspace{0.1cm} \mathrm{d} h \Big)}\bigg); & \\ h_{min} \leq z < h_{max} \\ 
\bigg(1- \exp{\Big(A \times  \int_{h_{min}} ^ {h_{max}} \mathcal{N}(z,h)\hspace{0.1cm} \mathrm{d} h  \Big)}\bigg); & \\ z \geq h_{max}
    \end{matrix} \right.
    \label{cdf_z2}
\end{equation}
Here $A$= $\frac{-2\lambda_u \pi}{(h_{max}- h_{min})}$ and $\delta(z,h)= \sqrt{z^2-h^2}$, for $R > 0$.

% \begin{equation}
% % \begin{align}
%     \mathcal{L}(z,h)= \frac{\exp{(-p\eta)} - L_S(z,h)(q\eta\sqrt{z^2-h^2}+1)}{(q\eta)^2}  \mathcal{N}(z,h)= \frac{(z^2-h^2)}{2} -\mathcal{L}(z,h)\label{bb1}
%    % \end{align}
% \end{equation}

\begin{equation}
\mathcal{L}(z,h)= \frac{\exp{(-p\eta)} - L_S(z,h)(q\eta\sqrt{z^2-h^2}+1)}{(q\eta)^2}
\label{bb1}
\end{equation}

\begin{equation}
\mathcal{N}(z,h)= \frac{(z^2-h^2)}{2} -\mathcal{L}(z,h)
% C(z,h)=\frac{\Big(q\sqrt{z^2-h^2} - \exp{(q\sqrt{z^2-h^2})}+1\Big) L_S(z,h)}{q^2} +\\ \frac{(z^2-h^2)}{2} 
  \label{cc1}
\end{equation}
\end{lemma}
\begin{proof}
    See Appendix A.
\end{proof}
\textbf{Remark:} If $R=0$, when the UAV is exactly above the user, $L_S(z,h)=1$. 
 % Moreover, (\ref{bb1}) and  (\ref{cc1}) will be replaced with $\mathcal{L}(z,h)= \frac{(z^2-h^2)}{2}$ and $\mathcal{N}(z,h)=0$.
 In that case, $\mathcal{L}(z,h)= \frac{(z^2-h^2)}{2}$, $\mathcal{N}(z,h)=0$.

$p$ and $q$ are the functions of intensity and dimensions of blockages as defined in (\ref{eq6}). $B_L$ is the probability that there exists at least one \ac{LoS} UAV-BS. 

Based on the blockage process considered in this paper, in high-blockage scenarios, The probability of having at least one \ac{LoS} \ac{UAV-BS} is not guaranteed to be one. Therefore, from \cite{34}, we derive the expression for $B_L$ given as
\begin{multline}
        B_L=1-\exp\Bigg( A \times \bigg[\int_{h_{min}}^{h_{max}} \int_{h_{min}}^{z_1} \mathcal{L}(z_1,h_1) \mathrm{d} h_1 \hspace{0.1 cm} \mathrm{d} z_1 + \\ \int_{h_{max}}^\infty\int_{h_{min}}^{h_{max}}\mathcal{L}(z_2,h_2) \mathrm{d}h_2 \hspace{0.1 cm} \mathrm{d}z_{2} \bigg] \Bigg).
        \label{bl1}
\end{multline}
However, given the blockage process, we observe that the probability of having at least one NLoS UAV-BS, considering they are distributed according to a \ac{PPP}, is always one.

Taking the derivative of (\ref{cdf_z1}) and (\ref{cdf_z2}) with respect to $z$, the \ac{PDF} of the distance of the typical user from the nearest LoS and NLoS UAV-BSs, given as, $f_{d_L}(z)$ and $f_{d_N}(z)$ respectively.
\begin{equation}
    f_{d_L}(z)= \left\{ \begin{matrix} f'_{d_L}(z); & h_{min} \leq z <h_{max}\\
    f^{''}_{d_L}(z); & z \geq h_{max} 
\end{matrix} \right . 
\end{equation}
Similar expressions are considered for $f_{d_N}(z)$.

Next, we derive the association probabilities for \ac{LoS} and \ac{NLoS} UAV-BSs. As previously mentioned, a typical user can connect to either a \ac{LoS} or \ac{NLoS} UAV-BS, depending on the maximum RSSI strategy.
\begin{lemma}
    The probability of association of the typical user with an \ac{NLoS} UAV-BS, $A_N$ is given by: 
% \begin{equation}
% A_N(d_N)=\exp{\bigg(A  \times \int_{h_a} \mathcal{L}(d_N^{\frac{\alpha_N}{\alpha_L}},h_a) \hspace{0.1cm} \mathrm{d}h_a \bigg)}
% % (1-B_L)+\\ \bigg[\exp{\bigg(A  \times \int_{h_{min}} ^ {h_{max}} \mathcal{L}(d_N^{\frac{\alpha_N}{\alpha_L}},h_a) \hspace{0.1cm} \mathrm{d}h_a \bigg)}-1+B_L\bigg]
% \label{ass1}
% \end{equation}

\begin{multline}  A_N=\int_{h_{min}}^{h_{max}}\bigg[\exp{\bigg(A  \times \int_{h_{min}} ^ {h_{max}} \mathcal{L}(r_1^{\frac{\alpha_N}{\alpha_L}},h_a) \hspace{0.1cm} \mathrm{d}h_a \bigg)}\bigg]  \\ f'_{d_N}(r_1) \mathrm{d} r_1 + \int_{h_{max}}^{\infty} \bigg[\exp{\bigg(A  \times \int_{h_{min}} ^ {h_{max}} \mathcal{L}(r_2^{\frac{\alpha_N}{\alpha_L}},h_a) \hspace{0.1cm} \mathrm{d}h_a \bigg)}\bigg] \\ f^{''}_{d_N}(r_2) \mathrm{d} r_2 
    \label{ass1}
\end{multline}
\end{lemma}
where $\mathcal{L}(z,h)$ is obtained from (\ref{bb1}) and $A$= $\frac{-2\lambda_u \pi}{(h_{max}- h_{min})}$.
% Marginalizing over $d_N$, we obtain the probability of \ac{NLoS} UAV-BS association, $A_N$.
\begin{proof}
See Appendix B
\end{proof}
Naturally, the probability of \ac{LoS} UAV-BS association $A_L= 1- A_N$. This solution is valid when $\alpha_N > \alpha_L$ and $h_{min} << h_{max}$, which is feasible in practical scenarios.

\section{ Meta Distribution of SINR Analysis}
\label{CSP_MD}
In this section, we derive the \ac{CSP} experienced by the typical user associated with the \ac{LoS} or \ac{NLoS} UAV-BS, conditioned on the point processes $\Phi_U$ and $\Phi_b$, to obtain the MD of SINR. Taking an expectation over the $\Phi_U$ and $\Phi_b$, we calculate the $b^{th}$ moment of the \ac{CSP}. Without loss of generality, leveraging the ergodicity property of the \ac{PPP}~\cite{2}, we conduct the downlink analysis from the perspective of a typical user at the origin.

\subsection{Conditional success probability of LoS UAV-BS}

The coverage probability or conditional success probability $P_{SL}(\gamma)$ experienced by the user associated to an \ac{LoS} UAV-BS, conditioning on $\Phi_U$ and $\Phi_b$, is given as
\begin{equation}
P_{SL}(\gamma)= \mathbb{P}\bigg(\frac{P_u K G_L d_L^{-\alpha_L}}{\sigma_N + I_L'+I_N} > \gamma | \Phi_U,\Phi_b\bigg),
     \label{csp_1}
  \end{equation}
  where $I_L'$ and $I_N$ are the interfering strengths from the other \ac{LoS} and \ac{NLoS} UAV-BSs respectively, where $I_L'= \sum_{i:\textbf{X}_i \in  \Phi_L'} P_u K G_L^{'} d_L'^{-\alpha_L}$ and $I_N= \sum_{i:\textbf{X}_i \in  \Phi_N }P_u K g_N d_N^{-\alpha_N}$. $\Phi_L'$ is the point process of LoS UAV-BSs in which the associated LoS UAV-BS is omitted.

 \begin{theorem}
\label{Th1}
 The $b^{th}$ moment of $P_{SL}(\gamma)$ is given as
 \begin{multline}
     M_{L}(\gamma)= A_L \int_{h_{min}}^{h_{max}} \Bigg [ A'(z) U'(z,\gamma,b) \times f_{d_L}(z)  \Bigg] \mathrm{d} z + \\ \int_{h_{max}}^\infty  \Bigg[ B'(z) U'(z,\gamma,b)   \times f_{d_L}(z)  \Bigg] \mathrm{d} z,
     \label{md1}
 \end{multline}
 where
\begin{equation}
 A'(z)=\exp{\Big(-2 \pi \lambda_u l_1'(z)\Big)} \times \exp{\Big(-2 \pi \lambda_u l_2'(z)\Big)}.
 \label{inn_1}
 \end{equation}
 \begin{equation}
 B'(z)=\exp{\Big(-2 \pi \lambda_u l_3'(z)\Big)} \times \exp{\Big(-2 \pi \lambda_u l_4'(z)\Big)}
  \label{inn_2}.
\end{equation}
 \begin{multline}
l_1'(z)= \int_0^{k(z)} \Bigg(1- \bigg[\int_{n(z,x)}^{h_{max}} \frac{\rho'(t,x)}{(h_{max}-n(z,x))}\mathrm{d}t\bigg]\Bigg)  x \mathrm{d}x + \\\int_{k(z)}^\infty \Bigg(1- \bigg[\int_{h_{min}}^{h_{max}} \frac{\rho'(t,x)}{(h_{max}-h_{min})} \mathrm{d}t\bigg]\Bigg) x \mathrm{d}x . 
 \end{multline}
  \begin{multline}
l_2'(z)= \int_0^{k(z)} \Bigg(1- \bigg[\int_{n(z,x)}^{h_{max}} \frac{\tau'(t,x)}{(h_{max}-n(z,x))}\mathrm{d}t\bigg]\Bigg)  x \mathrm{d}x + \\ \int_{k(z)}^\infty \Bigg(1- \bigg[\int_{h_{min}}^{h_{max}} \frac{\tau'(t,x)}{(h_{max}-h_{min})} \mathrm{d}t\bigg]\Bigg) x \mathrm{d}x . 
 \end{multline}
\begin{multline}
l_3'(z)= \int_{l(z)}^{k(z)} \Bigg(1- \bigg[\int_{n(z,x)}^{h_{max}} \frac{\rho'(t,x)}{(h_{max}-n(z,x))}\mathrm{d}t\bigg]\Bigg)  x \mathrm{d}x + \\ \int_{k(z)}^\infty \Bigg(1- \bigg[\int_{h_{min}}^{h_{max}} \frac{\rho'(t,x)}{(h_{max}-h_{min})} \mathrm{d}t\bigg]\Bigg) x \mathrm{d}x .  
 \end{multline}
 \begin{multline}
l_4'(z)= \int_{l(z)}^{k(z)} \Bigg(1- \bigg[\int_{n(z,x)}^{h_{max}} \frac{\tau'(t,x)}{(h_{max}-n(z,x))}\mathrm{d}t\bigg]\Bigg)  x \mathrm{d}x + \\ \int_{k(z)}^\infty \Bigg(1- \bigg[\int_{h_{min}}^{h_{max}} \frac{\tau'(t,x)}{(h_{max}-h_{min})} \mathrm{d}t\bigg]\Bigg) x \mathrm{d}x .  
 \end{multline}
 
$U'(z,\gamma,b)$ =$\exp{\bigg(\frac{-\gamma \sigma_N b }{P_u K z^{-\alpha_L}}\bigg)}$, $k(z)$=$\sqrt{z^2-h_{min}^2}$, $n(z,x)$=$\sqrt{z^2-x^2}$, $l(z)$=$\sqrt{z^2-h_{max}^2}$.
\begin{multline}
\rho'(t,x)=
L_S(\sqrt{x^2+t^2},t) \Big(\frac{1}{1+\frac{\gamma (\sqrt{x^2+t^2})^{-\alpha_L}}{z^{-\alpha_L}}}\Big)^b + \\ N_S(\sqrt{x^2+t^2},t) \Big(\frac{1}{1+\frac{\gamma (\sqrt{x^2+t^2})^{-\alpha_N}}{z^{-\alpha_L}}}\Big)^b. 
\label{eq123}
\end{multline}
\begin{multline}
\tau'(t,x)=
L_S(\sqrt{x^2+t^2},t) \Big(\frac{m}{m+\frac{\varepsilon \gamma (\sqrt{x^2+t^2})^{-\alpha_L}}{z^{-\alpha_L}}}\Big)^{mb} + \\ N_S(\sqrt{x^2+t^2},t) \Big(\frac{m}{m+\frac{\varepsilon \gamma (\sqrt{x^2+t^2})^{-\alpha_N}}{z^{-\alpha_L}}}\Big)^{mb}.
\label{eq234}
\end{multline}
\end{theorem}
where $\varepsilon= m(m!)^{\frac{-1}{m}}$.

\begin{proof}
    See Appendix C.
\end{proof}
% Substituting $b=1$, we obtain the SINR coverage probability experienced by the typical user associated with an \ac{LoS} UAV-BS.

\subsection{Conditional success probability of NLoS UAV-BS}

Here, we derive the conditional success probability $P_{SN}(\gamma)$ of the typical user associated to an \ac{NLoS} UAV-BS, i.e.,
\begin{equation}
P_{SN}(\gamma)= \mathbb{P}\bigg(\frac{P_u K g_N d_N^{-\alpha_N}}{\sigma_N + I_N'+I_L} > \gamma | \Phi_U,\Phi_b\bigg),
     \label{csp_3}
  \end{equation}
 where $I_L$ and $I_N'$ are the interfering strengths from the \ac{LoS} and the other \ac{NLoS} UAV-BSs respectively, where $I_L= \sum_{i:\textbf{X}_i \in  \Phi_L} P_u K G_L d_L^{-\alpha_L}$ and $I_N= \sum_{i:\textbf{X}_i \in  \Phi_N' }P_u K g_N' d_N'^{-\alpha_N}$. $\Phi_N'$ is the point process of NLoS UAV-BSs in which the associated NLoS UAV-BS is omitted.
 
% Conditioning on $\Phi_U$ and $\Phi_b$, applying the \ac{CCDF} of the exponentially random variable $g_n$ and taking the expectation over $g_n'$. 
% \begin{equation}
%      P'_{SN}(\gamma) = \exp{\bigg(\frac{-\gamma \sigma_N}{P_u K d_{N}^{-\alpha_N}}\bigg)} \prod_{i:\textbf{X}_i \in \Phi_U \backslash \textbf{X}_N}\bigg( \frac{1}{1+\frac{\gamma d'^{-\alpha_i}_{i}}{d_{N}^{-\alpha_N}}} \bigg).
%      \label{psn_1}
% \end{equation}
% where $\textbf{X}_i$ are the locations of the interfering \ac{NLoS} and all the \ac{LoS} UAV-BSs from the typical user. $\textbf{X}_N$ is the location of the associated \ac{NLoS} UAV-BS at a distance of $d_N$ from the typical user. By taking the expectation over $\Phi_U$ and $\Phi_b$ in (\ref{psn_1}), we solve the $b^{th}$ moment of $P_{SN}(\gamma)$.
\begin{theorem}
 The $b^{th}$ moment of $P_{SN}(\gamma)$ is given as
 \begin{multline}
     M_{N}(\gamma)= A_N \int_{h_{min}}^{h_{max}} \Bigg [ A^{''}(z) U^{''}(z,\gamma,b)   \times f_{d_N}(z)  \Bigg] \mathrm{d} z +\\  \int_{h_{max}}^\infty  \Bigg[ B^{''}(z) U^{''}(z,\gamma,b)    \times f_{d_N}(z)  \Bigg] \mathrm{d} z,
     \label{md2}
 \end{multline}
 where
\begin{equation}
 A^{''}(z)=\exp{\Big(-2 \pi \lambda_u l_1^{''}(z)\Big)} \times \exp{\Big(-2 \pi \lambda_u l_2^{''}(z)\Big)}.
 \label{inn_3}
 \end{equation}
 \begin{equation}
 B^{''}(z)=\exp{\Big(-2 \pi \lambda_u l_3^{''}(z)\Big)} \times \exp{\Big(-2 \pi \lambda_u l_4^{''}(z)\Big)}
  \label{inn_4}.
\end{equation}
 \begin{multline}
l_1^{''}(z)= \int_0^{k(z)} \Bigg(1- \bigg[\int_{n(z,x)}^{h_{max}} \frac{\rho^{''}(t,x)}{(h_{max}-n(z,x))}\mathrm{d}t\bigg]\Bigg)  x \mathrm{d}x + \\ \int_{k(z)}^\infty \Bigg(1- \bigg[\int_{h_{min}}^{h_{max}} \frac{\rho^{''}(t,x)}{(h_{max}-h_{min})} \mathrm{d}t\bigg]\Bigg) x \mathrm{d}x . 
 \end{multline}
  \begin{multline}
l_2^{''}(z)= \int_0^{k(z)} \Bigg(1- \bigg[\int_{n(z,x)}^{h_{max}} \frac{\tau^{''}(t,x)}{(h_{max}-n(z,x))}\mathrm{d}t\bigg]\Bigg)  x \mathrm{d}x + \\ \int_{k(z)}^\infty \Bigg(1- \bigg[\int_{h_{min}}^{h_{max}} \frac{\tau^{''}(t,x)}{(h_{max}-h_{min})} \mathrm{d}t\bigg]\Bigg) x \mathrm{d}x . 
 \end{multline}
\begin{multline}
l_3^{''}(z)= \int_{l(z)}^{k(z)} \Bigg(1- \bigg[\int_{n(z,x)}^{h_{max}} \frac{\rho^{''}(t,x)}{(h_{max}-n(z,x))}\mathrm{d}t\bigg]\Bigg)  x \mathrm{d}x + \\ \int_{k(z)}^\infty \Bigg(1- \bigg[\int_{h_{min}}^{h_{max}} \frac{\rho^{''}(t,x)}{(h_{max}-h_{min})} \mathrm{d}t\bigg]\Bigg) x \mathrm{d}x .  
 \end{multline}
 \begin{multline}
l_4^{''}(z)= \int_{l(z)}^{k(z)} \Bigg(1- \bigg[\int_{n(z,x)}^{h_{max}} \frac{\tau^{''}(t,x)}{(h_{max}-n(z,x))}\mathrm{d}t\bigg]\Bigg)  x \mathrm{d}x + \\ \int_{k(z)}^\infty \Bigg(1- \bigg[\int_{h_{min}}^{h_{max}} \frac{\tau^{''}(t,x)}{(h_{max}-h_{min})} \mathrm{d}t\bigg]\Bigg) x \mathrm{d}x .  
 \end{multline}
 
$U'(z,\gamma,b)$ =$\exp{\bigg(\frac{-\gamma \sigma_N b }{P_u K z^{-\alpha_N}}\bigg)}$, $k(z)$=$\sqrt{z^2-h_{min}^2}$, $n(z,x)$=$\sqrt{z^2-x^2}$, $l(z)$=$\sqrt{z^2-h_{max}^2}$.
\begin{multline}
\rho^{''}(t,x)=
L_S(\sqrt{x^2+t^2},t) \Big(\frac{1}{1+\frac{\gamma (\sqrt{x^2+t^2})^{-\alpha_L}}{z^{-\alpha_N}}}\Big)^b + \\
N_S(\sqrt{x^2+t^2},t) \Big(\frac{1}{1+\frac{\gamma (\sqrt{x^2+t^2})^{-\alpha_N}}{z^{-\alpha_N}}}\Big)^b.  
\end{multline}
\begin{multline}
\tau^{''}(t,x)=
L_S(\sqrt{x^2+t^2},t) \Big(\frac{m}{m+\frac{\varepsilon \gamma (\sqrt{x^2+t^2})^{-\alpha_L}}{z^{-\alpha_N}}}\Big)^{mb} +  \\ N_S(\sqrt{x^2+t^2},t) \Big(\frac{m}{m+\frac{\varepsilon \gamma (\sqrt{x^2+t^2})^{-\alpha_N}}{z^{-\alpha_N}}}\Big)^{mb}.  
\end{multline}

% where $\varepsilon= m(m!)^{\frac{-1}{m}}$.
\end{theorem}
The proof is the same as Theorem \ref{Th1}; therefore, we skip it for brevity. 

% \subsection{\ac{MD} of SINR}
The \ac{MD} of SINR provides a probabilistic framework to analyze the reliability of communication links in a wireless network.
For an ergodic point process in \ac{UAV} networks, the \ac{MD} $\bar{F}_{P_{S}}(\gamma, x_r)$ can be interpreted as the proportion of active links where the success probability $\mathcal{P}_S(\gamma)$ for a given value of threshold $\gamma$ exceeds the reliability threshold $x_r$. The exact \ac{MD} can be obtained from the Gil-Pelaez theorem~\cite{3} with $b^{th}$ moment $M_{b}$ of $\mathcal{P}_S(\gamma)$, $b \in \mathbb{R}$, $i = \sqrt{-1}$.
\begin{equation}
    \bar{F}_{P_{S}}(\gamma, x_r)= \frac{1}{2} + \frac{1}{\pi} \int_0^\infty \frac{\mathfrak{I}(e^{-it \log x_r } M_{it})}{t} \mathrm{d} t,
    \label{md}
\end{equation}
where $\mathfrak{I}(z)$ is the imaginary part of $z \in \mathbb{C}$. $M_{it}$ is the $it^{th}$ moment of $\mathcal{P}_S(\gamma)$. 

From (\ref{md1}) and (\ref{md2}), the $b^{th}$ moment of the  $\mathcal{P}_S(\gamma)$ given as
\begin{equation}
    M_b(\gamma)=A_L M_{L}(\gamma) +A_N M_{N}(\gamma).
    \label{eq9}
\end{equation}
The \ac{SINR} \ac{MD} is obtained by substituting (\ref{eq9}) in (\ref{md}).

The first moment of $\mathcal{P}_S(\gamma)$, substituting $b=1$ in (\ref{eq9}), gives the overall SINR coverage probability experienced by the user. The $(-1)^{th}$ moment of $\mathcal{P}_S(\gamma)$ is referred to as \ac{MLD} $M_{-1}(\gamma)$, represents the average number of retransmission attempts required to successfully transmit a packet between the UAV-BS and the user.

\section{Cache-Enabled HO Analysis}
\label{HO_AT}

% We are exploring a \ac{HO} management strategy aimed at minimizing handovers to provide a seamless user experience and less latency across applications such as autonomous driving, navigation, and video streaming. According to 3GPP standards~\cite{33}, users perform a cell search every $t_s$ seconds. A handover is triggered when a neighboring UAV-BS offers a \ac{RSS} higher than the associated UAV-BS. Each user requires
% $\Delta T$ seconds to measure this \ac{RSS} from the neighboring cell.
% We dynamically select $t_s$ by leveraging caching capabilities in \ac{UE} in motion. We consider every user to have a rate requirement of $s_r$, and every \ac{UE} has a cache size of $G$ and moving with a velocity of $v$. 
% % From the \ac{UE} respective, $H_N$ is the number of \ac{HO} performed in one realization of the point process, and $R_a$ is the average throughput experienced by the user in one realization. 
% For a total time of the travel $T$, the analysis is performed at each coherence time $t_c$. Initially, the cached data $C_D$ in the cache memory is zero.

In this section, we discuss an efficient \ac{HO} management scheme by leveraging the caching capabilities of the \ac{UE}. In our work in \cite{9}, we discuss HO management using caching in a \ac{1-D} network and analyze the impact of handovers on network performance and delay. Here, we extend this to a UAV-based network, with locations of \acp{UAV} modeled as a 2-D \ac{MPPP}, in the presence of blockages. This HO management scheme reduces handovers in the network, thus reducing the latency in transmission in the presence of mobile users. We perform a semi-analytical analysis of the network to derive the average throughput experienced by the user in a cache-enabled HO scheme.
Without loss of generality, we perform a downlink analysis from the perspective of a mobile user
moving in a straight line through the origin along the x-axis. 

We consider a scenario where a user moves over a long period, denoted as $T$. At $t=0$, the user will be associated with a UAV-BS at $(x_{a},y_{a},h_{a})$ and experience a download rate of $d_r(t)$. The download rate experienced by the user associated with an \ac{LoS} \ac{UAV-BS} at time $t$ is given as
\begin{equation}
d_r(t)= B \log_2\Big(1+ \frac{P_u K G_L d_{mu}(t)^{-\alpha_L}}{\sigma_N + I_L'+I_N}\Big),
\end{equation}
where $I_L'$ and $I_N$ are the interfering strengths from the other \ac{LoS} and \ac{NLoS} UAV-BSs respectively. $d_{mu}(t)$ is the distance between the associated \ac{BS} and the user at time $t$, given as $
    d_{mu}(t)= \sqrt{(x_{a}-vt)^2 + (y_{a}-y_{u})^{2}+ h_{a}^{2}}$.
    
    We assume the user is moving along the \ac{2-D} plane in a straight line with velocity $v$. $G_L$ is the Nakagami-m fading factor, $K$ is the pathloss coefficient i.e., $K=(\frac{(\lambda_c)}{4\pi})^2$, where $\lambda_c$ is the carrier wavelength. If it is \ac{NLoS} \ac{UAV-BS}, the path-loss and fading parameters change.

Consider the scenario where $d_r > s_r$. As the user moves forward, the user experiences a download rate greater than 
$s_r$. Consequently, the excess data (i.e., the difference between the downloaded data and the required data for the service) is cached in the \ac{UE} memory of cache size $G$. The caching continues as long as the user's download rate is equal to or exceeds the required service rate and there is available space in the cache. Let 
$t_d$ represent the time at which the download rate falls below the service's required rate.  This time, referred to as the caching time, marks when caching occurs at the \ac{UE}, given the download rate $d_r$ is greater than $s_r$. In the presence of blockages, this time is discontinuous, making its analytical expression intractable. The downloaded data till time $t_d$ is given as $C_D=\int_0^{t_d} d_r(t) \hspace{0.1cm} \mathrm{d}t.$

The amount of data used by the user to meet the service requirement till time $t_n$ is given as $s_r t_n$. The user requires
$\Delta T$ seconds to measure the \ac{RSSI} from the neighboring cells. If there is sufficient data cached, i.e., more than $s_r \Delta T$, to skip \acp{HO} and perform \ac{RSS} measurements, the cell search is muted, and \acp{HO} are skipped. The user remains associated with the previous \ac{UAV-BS} until the next cell search is initiated. While utilizing cached data $C_D$, the user achieves the specified service rate $s_r$. Let the time be $t_n$, which the user can use the cached data to maintain the service requirement rate. 
If there is $C_D$ data in the cache, the next cell search is triggered when the condition $C_D/s_r < \Delta T$ is met. The user maintains association with the current \ac{UAV-BS} until this condition is satisfied or the download rate falls below the minimum threshold defined by 3GPP to sustain a connection. Once this occurs, the service rate drops to the current download rate $d_r(t)$, prompting a cell search, \ac{RSS} measurements, and the initiation of a handover. 
% This is given in the line 19 of the algorithm. If the index of the associated UAV-BS is changed, we can say a HO is initiated. 
Using this methodology, the number of handovers experienced by the user in the cache-enabled HO scheme is given as $H_N$. By performing temporal and spatial averaging, the average HO rate is given as $\mu= \frac{H_N}{T}$. Performing spatial averaging, the average HO rate is given as $\mu'$. The throughput experienced by the user when utilizing the cached data is $s_r$. Considering the entire time of travel $T$, the average throughput at the time of utilizing cached data is given as
$\frac{s_r t_n}{T}$. Performing spatial averaging, the average throughput experienced by the user when utilizing the cached data is given as $R_s$. While the user experiences the download rate $d_r$ less than the service rate, we derive the average throughput by utilizing the first moment of \ac{CSP} $M_1(\gamma)$ derived in (\ref{eq9}).
\begin{equation}
    R_d=B \int_0^{2^{\frac{s_r}{B}}-1} \frac{M_1(k)}{k+1} \mathrm{d} k.
\end{equation}
Here, the maximum SINR experienced by the user will be $2^{\frac{s_r}{B}}-1$, which is the SINR experienced by the user when the cached data is fully exhausted before initiating the cell search.

Therefore, the overall average throughput experienced by the user is given as $R_a'= R_s+R_b$.

% The average throughput $R_a'$ at which the user achieves the service requirements is given as
% \begin{equation}
%     R_a'= \frac{\sum_{t=0}^T (R(t))}{T/t_c}
% \end{equation}
% The average \ac{HO} rate is $\mu=\frac{H_N}{T}$ where $H_N$ is the number of \acp{HO} triggered by the user. 

% \begin{algorithm}[h]% enter the algorithm environment
% \caption{Cache-based HO Skipping}
% \footnotesize
% \textbf{One Spatial Realization}\\
% \textbf{Inputs}:\\
% {$\textbf{X}_u = [\mathbf{x}_u,\mathbf{y}_u,\mathbf{h}_u]$ = Dimensions of UAV-BSs}\\
% {$G$= Maximum Cache size at UE}\\
% {$p$= Rate required for the service}\\
% {$\Delta T$= RSS measuring time}\\
% {$t_c$= Coherence time}\\
% \textbf{Output}:\\
% {$H_N$ = Number of Handovers in one realization}\\
% {$R$ = Service Rate achieved in one realization}\\
% \textbf{Initialize}:\\
% {$C_D=0$, $H_N=0$, $K=0$}\\
% \label{alg:cap}
% \begin{algorithmic}
% [1]% enter the algorithmic environment
% \raggedright
% \FOR{$t=0:t_c:T$}
% \IF{$K=0$}
% \STATE{($x^{a}_u,y^{a}_u,h^{a}_u$)= $\textbf{X}_u(\mathrm{ind}(t))$}
% \STATE{$K=1$}
% \ENDIF
% \IF{($d_r(t) > p)$}
% \IF{$(C_D \leq G)$}
% \STATE{$C_D=C_D+ (D_d-pt_n$)}\\
% \STATE{$R(t)= p$}
% \ELSE
% \STATE{$R(t)= p$}
% \ENDIF
% \ELSE
% \IF{$(C_D  \geq p \Delta T)$}
% \STATE{$C_D=C_D- p t_c$}\\
% \STATE{$R(t)=p$}
% \ELSE
% \STATE{$R(t)= d_r(t)$}\\
% \IF{$\mathrm{ind}(t) \neq \mathrm{ind}(t-1)$}
% \STATE{$H_N=H_N+1$}\\
% \STATE{$K=0$}
% \ENDIF
% \ENDIF
% \ENDIF
% \ENDFOR\\
% \end{algorithmic}
% \end{algorithm}

\subsection{Effective Average Rate}
The effective average throughput experienced by the user is given as~\cite{9}
\begin{equation}
R_{\rm{eff}}=R_{a}'(1-\mu' t_H)^+ ,
\end{equation}
where $\mu'$ is the average handover rate. The time overhead for each handover $t_H$ is 43 ms~\cite{8}. The delay experienced by the user due to the HOs is expressed as $\mu' t_H$.
% Further, we perform an average over the spatial realizations to obtain the overall throughput experienced by the user to achieve the service rate requirement. 

\section{Numerical Results and Discussions}
\label{NM_D}

In this section, we validate the theoretical model through Monte Carlo simulations and provide numerical results to discuss the key characteristics of the network. 
% The simulation parameters are shown in Table \ref{TABLE: table11}.
The parameters are $P_{u}$= 10 W~\cite{10}, $h_{min}=100$ m, $h_{max}=300$ m, $B$= 100 MHz~\cite{11}, $\alpha_L=2$, $\alpha_N=4$, $v$= 1 m/s, $f_c= 3.5$ GHz, $\lambda_b=10^{-6}$ km$^{-2}$, $s_r=40$ Mbps, $t_H= 43$ ms, $T= 5000$ s, $t_s=20$ ms.

% \renewcommand{\arraystretch}{1}
% \begin{table}[]
%     \centering
% \caption{Simulation Parameters}
%     \label{TABLE: table11}
%     \begin{tabular}{|M{1cm}|M{6cm}|M{2.5cm}|}\hline
%     Notation & Parameter & Value \\\hline\hline
%     $P_u$ & Transmit power & 10 W~\cite{10} \\\hline
%      $\lambda_u$ & Intensity of UAV-BSs & 10$^{-5}$ Km$^{-2}$\\\hline
%      $\lambda_b$ & Intensity of blockages & 10$^{-6}$ Km$^{-2}$\\\hline
%      $h_{min}$ & Minimum altitude & 100 m\\\hline
%      $h_{max}$ & Maximum altitude & 300 m\\\hline
%      $H_b$ & Average blockage height & 80 m\\\hline
%      $B$ & Bandwidth &100 MHz~\cite{11}\\\hline
%      $\alpha_L$ & \ac{LoS} Path-loss Exponent & 2 \\\hline
%      $\alpha_N$ & \ac{NLoS} Path-loss Exponent & 4 \\\hline
%      $f_c$ & Carrier frequency & 3.5 GHz\\\hline
%      $\sigma^2$ & Noise density & -174 dBm/Hz~\cite{25}\\\hline
%      $s_r$ & Service rate requirement & 40 Mbps\\\hline
%      $v$ & Velocity & 1 m/s\\\hline
%      $t_H$ & HO overhead time & 43 ms\\\hline
%      $t_s$ & Search time & 20 ms\\\hline
%      $T$ & Total time & 5000 s\\\hline
%      \end{tabular}
%    \end{table}
   
In Fig.~\ref{MD_1}, we plot the \ac{SINR} \ac{MD} as a function of the intensity of \acp{UAV-BS} across various blockage scenarios, with a reliability threshold of $x_r = 0.9$, indicating that 90\% of the communication links are active. In low-blockage scenarios, for an intensity of $\lambda_u = 10^{-6}$, the \ac{SINR} MD value is 0.7. This implies that the probability of achieving the reliability threshold, where 90\% of the communication links between the user and the \acp{UAV-BS} are active, is 70\%. However, in scenarios with higher blockages, users achieve the reliability threshold of 90\% active links less than 55\% of the time. Contrary to popular belief, regions with higher blockages benefit from a denser deployment of \acp{UAV-BS}, as the blockages help reduce interference from other \acp{UAV-BS}. We also observe that increasing the network density does not improve performance beyond a certain limit. \begin{figure}[htbp]
   \begin{minipage}[t]{0.45\textwidth}
    \centering
    \includegraphics[width=1\linewidth,height=0.8\linewidth]{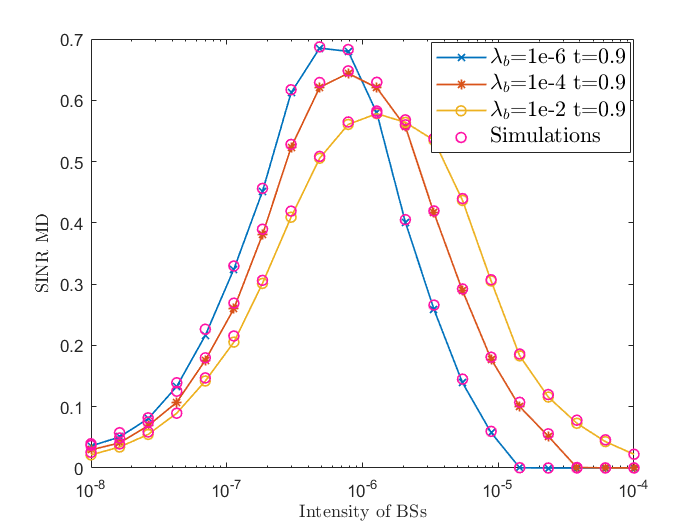}
    \caption{SINR MD as a function of the intensity of UAV-BSs for different blockage intensities for $x_r=0.9$, $\gamma=-10$ dB}
    \label{MD_1}
    \end{minipage}\hfill
    \begin{minipage}[t]{0.45\textwidth}
     \centering
    \includegraphics[width=1\linewidth]{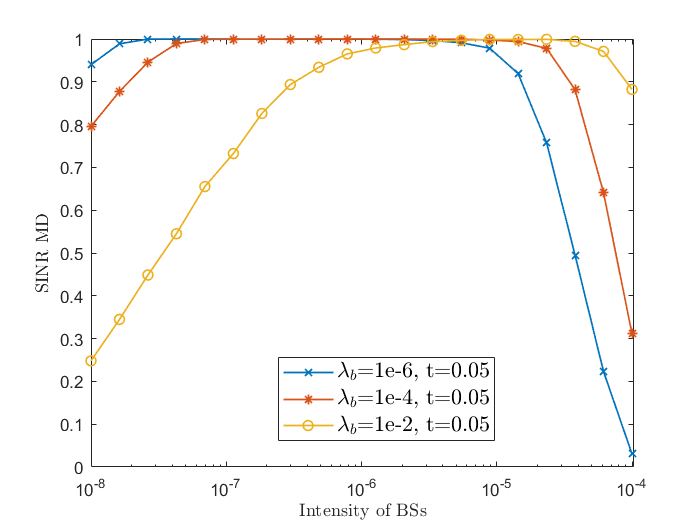}
    \caption{SINR MD as a function of the intensity of UAV-BSs for different blockage intensities for $x_r=0.05$, $\gamma= 10$ dB}
    \label{MD_2}
    \end{minipage}
    \end{figure} In Fig.~\ref{MD_2}, we plot the \ac{SINR} \ac{MD} as a function of the intensity of \acp{UAV-BS}, evaluated at a reliability threshold of 0.05. This threshold signifies that at least 5\% of the communication links are active, with the \ac{SINR} exceeding a specified value. The results reveal that areas with lower blockage require fewer \acp{UAV-BS} to sustain this level of active communication links, whereas higher blockage areas demand a denser deployment. In low-blockage scenarios where $\lambda_b=10^{-6}$, users consistently experience at least 5\% of active communication links. In contrast, in high-blockage scenarios $\lambda_b=10^{-2}$, increasing the density of \acp{UAV-BS} beyond 10 UAV-BSs per km$^{2}$ ensures that all users achieve at least 5\% of active communication links at all times.   

% \begin{figure}[htbp]
% \centering
% \includegraphics[width=1\linewidth,height=0.85\linewidth]{meta_0.05.eps}
% \caption{SINR MD as a function of the intensity of UAV-BSs for different blockage intensities for $x_r=0.05$}
% \label{MD_2}
% \end{figure}

 \begin{figure}[htbp]
   \begin{minipage}[t]{0.45\textwidth}
    \centering
    \includegraphics[width=1\linewidth,height=0.8\linewidth]{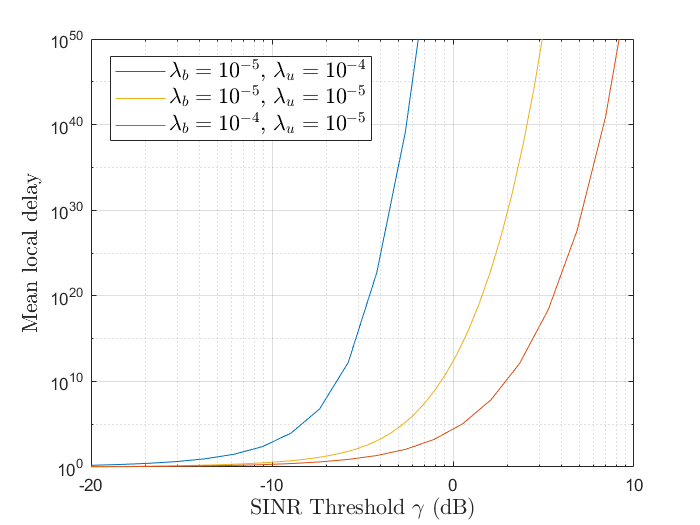}
    \caption{Mean local delay versus SINR threshold}
    \label{MLD}
    \end{minipage}\hfill
    \begin{minipage}[t]{0.45\textwidth}
     \centering
    \includegraphics[width=1\linewidth]{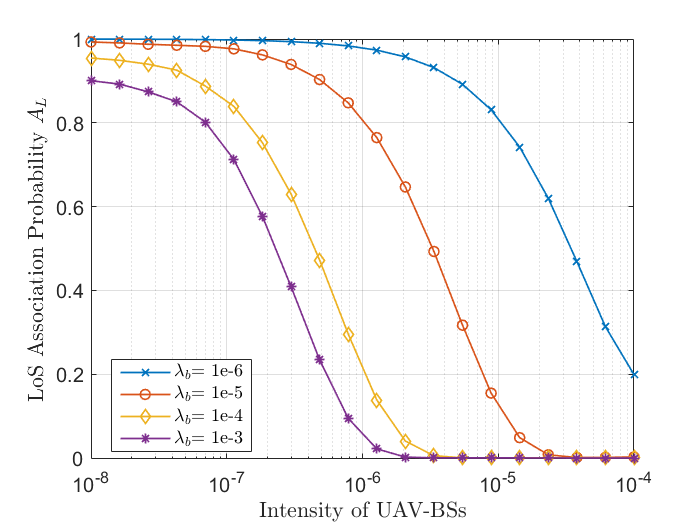}
    \caption{LoS Association Probabilities for different blockage intensities}
    \label{LOS_1}
    \end{minipage}
    \end{figure}

In Fig.~\ref{MLD}, we plot \ac{MLD} in the network for different values of SINR requirement at the users. Here, we observe for a fixed value of UAV-BS intensity of $10^{-5}$, the increase in blockage intensity leads to an increase in SINR threshold $\gamma$ because in order to maintain a consistent mean local delay while ensuring SINR constraints are met, the system must increase the SINR threshold. However, for fixed blockage intensity of $10^{-5}$, an increase in UAV-BS intensity leads to a decrease in $\gamma$. This is because the user can experience the required communication quality more easily due to better coverage.
% \begin{figure}[htbp]
% \centering
% \includegraphics[width=1\linewidth,height=0.85\linewidth]{MLD_combined_1.eps}
% \caption{Mean local delay versus SINR threshold}
% \label{MLD}
% \end{figure}
Fig.~\ref{LOS_1} shows the \ac{LoS} association probability of the network versus the intensity of UAV-BSs. As the intensity of UAV-BSs increases, \ac{LoS} association probability decreases. For sparse deployment of UAV-BSs, \ac{LoS} association is high and then decreases as the intensity of UAV-BSs increases. In the presence of \acp{TBS} as discussed in \cite{7}, \ac{LoS} associations are fewer for lower values of intensity of BSs. Further increasing the intensity of UAV-BSs enhances the likelihood of serving \ac{NLoS} \acp{UAV-BS} without significantly impacting the \ac{LoS} links. Here, we discuss the LoS association trend in the absence of \acp{TBS}. The key insight is that network densification alone is not an effective solution for increasing the LoS association probability or enhancing the network's coverage performance. Intuitively, as blockages increase, the probability of \ac{LoS} association decreases. However, a higher density of NLoS \acp{UAV-BS} can mitigate interference, leading to an improvement in \ac{SINR}, as discussed in Fig.~\ref{MD_1} and Fig.~\ref{MD_2}.

% \begin{figure}[htbp]
% \centering
% \includegraphics[width=1\linewidth,height=0.85\linewidth]{LoS_asso1.eps}
% \caption{LoS Association Probabilities for different blockage intensities}
% \label{LOS_1}
% \end{figure}

 \begin{figure}[htbp]
   \begin{minipage}[t]{0.45\textwidth}
    \centering
    \includegraphics[width=1\linewidth,height=0.8\linewidth]{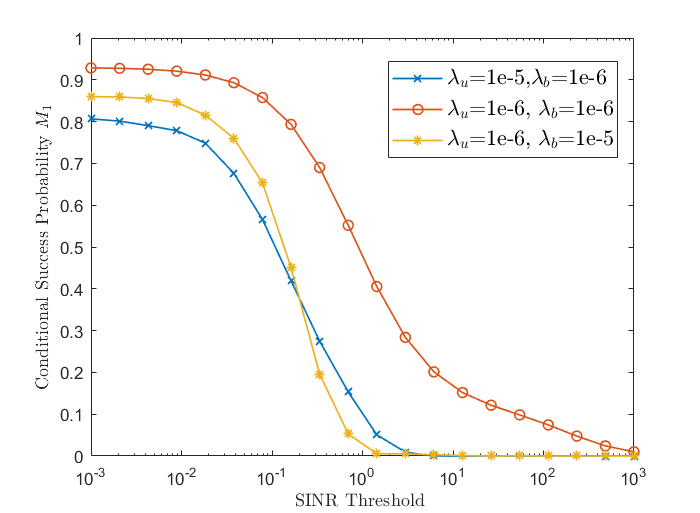}
    \caption{Conditional success probability versus \ac{SINR} threshold $\gamma$}
    \label{CSP_1}
    \end{minipage}\hfill
    \begin{minipage}[t]{0.45\textwidth}
     \centering
    \includegraphics[width=1\linewidth,height=0.8\linewidth]{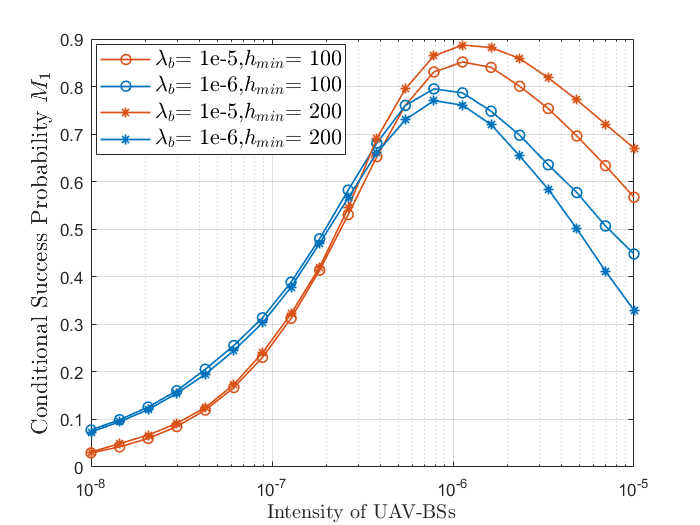}
    \caption{Conditional success probability of the network versus the intensity of UAV-BSs for different blockage intensities and UAV altitudes with $h_{max}=300$ m, $\gamma= -10$ dB}
    \label{CSP_2}
    \end{minipage}
    \end{figure}

In Fig.~\ref{CSP_1}, we observe that as the \ac{SINR} threshold increases, the \ac{CSP} decreases. For a given blockage intensity, increasing the deployment density of UAV-BSs degrades network performance due to an increase in interference. Furthermore, when user requirements are low, blockages do not significantly impact network performance, as they help block interference. However, when user requirements are high, blockages have a negative effect on performance by reducing \ac{LoS} connections, thereby impairing network reliability.  
% \begin{figure}[htbp]
% \centering
% \includegraphics[width=1\linewidth,height=0.85\linewidth]{SINR_prob_gamma.eps}
% \caption{Conditional success probability versus \ac{SINR} threshold $\gamma$}
% \label{CSP_1}
% \end{figure}
In Fig.~\ref{CSP_2}, we plot \ac{CSP} versus intensity of UAV-BSs for different blockage intensities and minimum altitude of the UAV-BSs. As discussed, in denser deployments, higher blockage scenarios offer better performance because of the blocking of the signals interfering \acp{UAV-BS}. Additionally, we observe that in areas with fewer blockages, deploying UAV-BSs at lower altitudes enhances performance by increasing \ac{LoS} connections. In contrast, for areas with higher blockage intensity, deploying \acp{UAV-BS} at higher altitudes improves performance by establishing more \ac{LoS} links, which is critical in overcoming the effects of large blockages.

% \begin{figure}[htbp]
% \centering
% \includegraphics[width=1\linewidth,height=0.85\linewidth]{SINR_hmin_final.eps}
% \caption{Conditional success probability of the network versus the intensity of UAV-BSs for different blockage intensities and UAV altitudes}
% \label{CSP_2}
% \end{figure}

Fig.~\ref{HO_1} shows the \ac{HO} delay is analyzed with respect to the intensity of \acp{UAV-BS} for various handover strategies. The conventional HO scheme exhibits higher handover rates compared to the cache-based scheme, as the latter bypasses unnecessary handovers by leveraging cached data at the \ac{UE}. The cache-based strategy avoids initiating new associations when the cached data at the \ac{UE} is utilized, even in the presence of blockages. This reduces the overall handover delay in the network, thereby decreasing latency. However, if we increase the intensity of \acp{UAV-BS} after a certain limit, the cache-based scheme results in the same HO delay as the conventional because of the decrease in the cache data at the UE end. This reduction occurs because increased BS density leads to higher interference, lowering the user’s download rate and, consequently, reducing the cached data. 
\begin{figure}[htbp]
   \begin{minipage}[t]{0.45\textwidth}
    \centering
    \includegraphics[width=1\linewidth,height=0.8\linewidth]{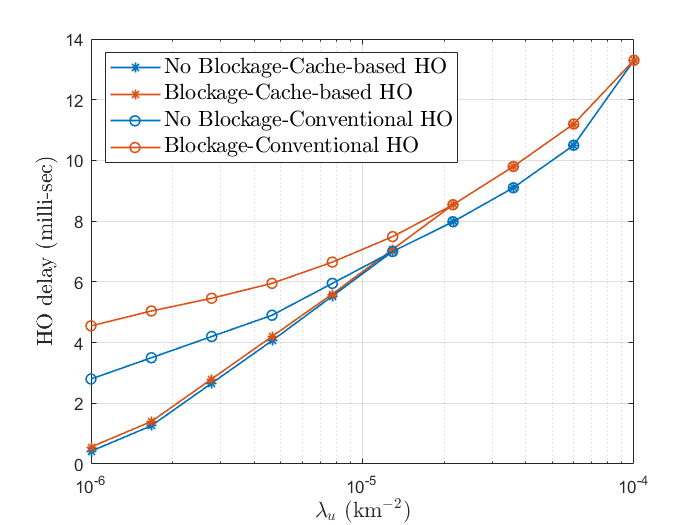}
    \caption{Handover Delay versus Intensity of UAV-BS with $s_r= 40$ Mbps, $\lambda_b= 10^{-6}$ km$^{-2}$}
    \label{HO_1}
    \end{minipage}\hfill
    \begin{minipage}[t]{0.45\textwidth}
     \centering
    \includegraphics[width=1\linewidth,height=0.8\linewidth]{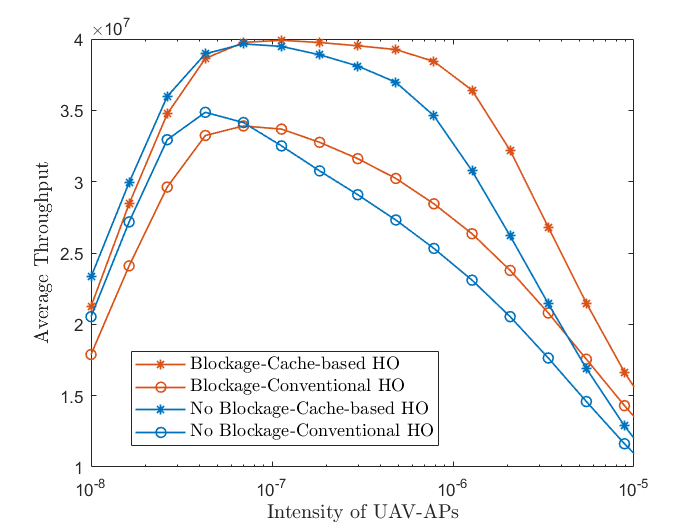}
    \caption{Effective Average Throughput versus Intensity of UAV-BS with $s_r= 40$ Mbps, $\lambda_b= 10^{-6}$ km$^{-2}$.}
    \label{HO_2}
    \end{minipage}
    \end{figure}
% \begin{figure}[htbp]
% \centering
% \includegraphics[width=1\linewidth,height=0.85\linewidth]{HO_delay.eps}
% \caption{Handover Delay versus Intensity of UAV-BS}
% \label{HO_1}
% \end{figure}
In Fig.~\ref{HO_2}, we plot the average throughput experienced by the UE versus the intensity of UAV-BSs. We observe that the cache-based scheme provides better performance compared to the conventional scheme. For sparse deployment of BSs, both the cache-based and conventional schemes are negatively affected by blockages. However, for dense deployments of UAV-BSs, blockages block the interference, thus improving the network performance. Nevertheless, after reaching a certain \ac{UAV-BS} deployment density, the average throughput for the cache-based scheme converges with that of the conventional scheme due to increased interference, leading to a depletion of the cached data. 
% \begin{figure}[htbp]
% \centering
% \includegraphics[width=1\linewidth,height=0.85\linewidth]{avg_rate_comb.eps}
% \caption{Effective Average Throughput versus Intensity of UAV-BS}
% \label{HO_2}
% \end{figure}

\begin{figure}[htbp]
\centering
\includegraphics[width=1\linewidth,height=0.8\linewidth]{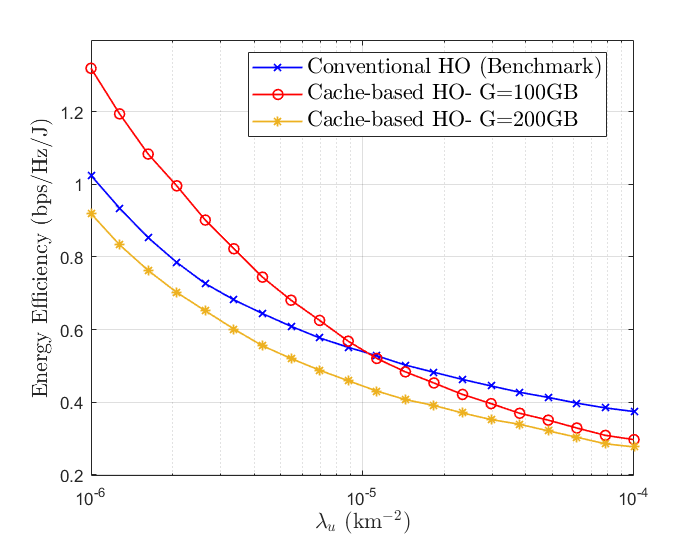}
\caption{Energy Efficiency versus Intensity of UAV-BSs with $\lambda_b= 10^{-6}$ km$^{-2}$, $s_r= 40$ Mbps}
\label{EE}
\end{figure}

In Fig.~\ref{EE}, we illustrate the variation in the energy efficiency of the system as the intensity of UAV-BSs increases, comparing conventional and cache-enabled handover (HO) schemes. The same power consumption model as in our previous work \cite{9} is adopted. For a caching capacity of 100 GB, the cache-enabled HO scheme outperforms the conventional scheme up to a UAV-BS intensity of approximately 10 UAV-BSs per $km^{2}$. This is due to the higher spectral efficiency and reduced inter-frequency \acp{HO} achieved through caching, even when accounting for the power consumption due to caching.
However, as the UAV-BS intensity increases further, higher interference levels result in a decline in spectral efficiency and an increase in the HO rate, ultimately reducing energy efficiency. For a larger cache size of 200 GB, due to the power consumption associated with caching, the cache-enabled scheme performs worse than the conventional scheme in terms of energy efficiency.  

\section{Conclusion}
\label{CL}
In this work, we studied a UAV-based network designed to meet the service rate requirements of mobile users in the presence of urban blockages. Using stochastic geometry, we modeled the network to derive the distance distribution and association probabilities for the nearest \ac{LoS} and \ac{NLoS} \acp{UAV-BS}. We observed increasing network density does not necessarily enhance \ac{LoS} associations, and blockages can improve network reliability by reducing interference. We analyze the \ac{MD} of \ac{SINR} to evaluate the network's reliability under various blockage scenarios. Additionally, we proposed a cache-enabled \ac{HO} management scheme that leverages caching at user equipment to minimize unnecessary handovers, reduce latency, and improve throughput. A complete analytical framework for the same will be addressed in the future.

% \section{Acknowledgement}
% This research is funded by the IIT Palakkad Technology IHub Foundation Doctoral Fellowship IPTIF/HRD/DF/026.

\begin{appendices}
\section{Proof of Lemma 2}
The heights of the \acp{UAV} are uniformly distributed from $h_{min}$ and $h_{max}$. $d_L$ and $d_N$ are the 3-D Euclidean distance of the nearest $\ac{LoS}$ and $\ac{NLoS}$ UAV-BSs from the typical user, respectively. The cdf of $d_L$ between the typical user and nearest \ac{LoS} UAV-BS is given as
\begin{equation}
    F_{d_L}(z)= \mathbb{P}(d_L \leq z)=1-\mathbb{P}(d_L > z)
\end{equation}
The two cases are: (i) $h_{min} \leq d_L < h_{max}$ (ii) $d_L \geq h_{max}$.

For case 1, we consider 
$h_{min} \leq d_L < h_{max}$, hence the heights of the UAV-BS is constrained to  [$h_{min}$, $d_L$].  Consequently, we consider exclusively on those UAV-BSs whose heights fall within this specified range, leading to the formulation of a thinned \ac{PPP}. The intensity of thinned PPP is given as
$\lambda_u \frac{d_L-h_{min}}{h_{max}-h_{min}}$, since the heights of UAV-BSs are uniformly distributed from $h_{min}$ to $h_{max}$. This models the spatial distribution of UAV-BSs within the specified height constraints.
The 3-D Euclidean distance $d_L$ is represented as 2-D distance $r_L$ and height $h_L$, where $d_L= \sqrt{r^2+h^2}$.
Given that a UAV-BS is located at a 3-D distance $d_L$ and height $h$, the probability that this UAV-BS is in line-of-sight (LoS) is denoted by $L_S(d_L,h)$. To determine the total probability of LoS with respect to the 2-D distance $r$ between the user and the UAV-BS, the analysis requires integrating the dimensions of
$L_S(\sqrt{r^2+h^2},h)$ over the polar coordinates 
$r$ and $\theta$ covering a circular region of radius 
$r$. This gives the 2-D area where the probability of a UAV-BS being in LoS is taken into account within the circular region. To account for the variations in $h$, which is uniformly distributed from $h_{min}$ to $h_{max}$, an expectation over $h$ is taken. This provides the effective area within the 3-D space where there is a possibility of the UAV-BS being in \ac{LoS} from the user. This quantifies the spatial region in which LoS conditions are likely to be met, considering both the 2-D distribution of locations of \acp{UAV-BS} and their random heights.

Applying the proper limits for $r$, $h$ and $\theta$, the void probability is given as
\begin{multline}
    \mathcal{V}(d_L)= \exp\bigg(-2 \pi \lambda_{u} \frac{(d_L-h_{min})}{(h_{max}-h_{min})} \\ \int_{h_{min}}^{d_L} \frac{1}{(d_L-h_{min})}\int_0^{\delta(d_L,h)} L_S(\sqrt{r^2+h^2},h) r \mathrm{d} r \hspace{0.1cm} \mathrm{d} h \bigg)
\end{multline}
where $\delta(d_L,h)={\sqrt{d_L^2-h^2}}$. The inner integral can be represented as $\mathcal{L}(d_L,h)= \int_0^{\delta(d_L,h)} L_S(\sqrt{r^2+h^2},h) r \mathrm{d} r$.

Therefore,  the CDF is given as $F_{d_L}(z)= \bigg(1-\mathcal{V}(z)\bigg)/B_L$.

$B_L$ is the probability that there is at least one \ac{LoS} UAV-BS given in (\ref{bl1}). The derivation follows the same as in \cite{34}.

For case 2, where $d_L$ is greater than or equal to $h_{max}$, the heights of UAV-BSs can be anywhere between $h_{min}$ and $h_{max}$. 
% Therefore, the probability that there are \ac{LoS} UAV-BSs within the distance of $d_L$ and height $h$ is given as $L_S(d_L,h) \mathbb{P}(h_{min} \leq h \leq h_{max})$. 
Given that the altitudes of the UAVs are uniformly distributed from $h_{min}$ to $h_{max}$,  $\mathbb{P}(h_{min} \leq h \leq h_{max}) = 1$. Consequently, the entire spatial distribution of UAV-BS within the height of $h_{min}$ to $h_{max}$ are considered. By calculating the effective area where there is a possibility of being in \ac{LoS} from the user, considering the 2-D locations and random heights, the CDF $F_{d_N}(z)$ can be represented as
\begin{multline}
   F_{d_N}(z)= \bigg[1-\exp\bigg(-2 \pi \lambda_{u} \int_{h_{min}}^{h_{max}} \frac{1}{(h_{max}-h_{min})}\\    \int_0^{\delta(z,h)} L_S(\sqrt{r^2+h^2},h) r \mathrm{d} r \hspace{0.1cm} \mathrm{d} h \bigg)\bigg]/B_L
\end{multline}

Similarly, the distance distribution of nearest \ac{NLoS} UAV-BS $F_{d_N}(z)$ can be derived. If $R$ = 0, then $\eta$ = 0, and the inner integral $\int_0^{\delta(z,h)} L_S(\sqrt{r^2+h^2},h) r \mathrm{d} r \hspace{0.1cm} \mathrm{d} h = \frac{(z^2-h^2)}{2}$.

\section{Proof of Lemma 3}
The probability of associating to an \ac{NLoS} UAV-BS can be derived as
\begin{gather}
    \mathbb{P}(P_u d_N^{-\alpha_N} > P_u d_L^{-\alpha_L})
    % \mathbb{P}(d_N < d_L^{\frac{\alpha_L}{\alpha_N}})
    = \mathbb{P}(d_L > d_N^{\frac{\alpha_N}{\alpha_L}})
    \end{gather}
The event in which there is at least one \ac{LoS} UAV-BS is given as $E_1$. The event that there are no \ac{LoS} UAV-BSs is given as $E_2$. The event that there is atleast one \ac{NLoS} UAV-BSs is given as $E_3$. Based on our blockage model, $\mathbb{P}(E_3)=1$.

The probability that the other events occur is given below:
\begin{inparaenum}[(i)]
   \item $\mathbb{P}(E_1)= B_L$
    \item $\mathbb{P}(E_2)= (1-B_L)$.
\end{inparaenum}
$B_L$ is given in (\ref{bl1}).
% \begin{multline}
%    B_N=1-\exp\bigg( A \times \Big(\int_{h_{min}}^{h_{max}} \int_{h_{min}}^{z_1} C(z_1,h_1) \mathrm{d} h_1 \hspace{0.1 cm} \mathrm{d} z_1 + \\ \int_{h_{max}}^\infty\int_{h_{min}}^{h_{max}} C(z_2,h_2) \mathrm{d}h_2 \hspace{0.1 cm} \mathrm{d}z_{2} \Big) \bigg).
% \end{multline}
% where $C(z,h)$ is given in (\ref{cc1}).

Therefore, the user getting associated with an \ac{NLoS} UAV-BS can occur in two possible scenarios:
\begin{itemize}
    \item    
    The joint probability of having no LoS UAV-BSand at least one NLoS UAV-BS, the probability of associating to NLoS UAV-BS is given as $   \mathbb{P}(E_3)\mathbb{P}(E_2)$.    
    \item The joint probability of having at least one LoS UAV-BS and at least one NLoS UAV-BS, and under those conditions, the probability of associating to NLoS UAV-BS is given as
    \begin{equation}
       \mathbb{P}(E_3)\mathbb{P}(E_1) \Big[\mathbb{P}(d_L > d_N^{\frac{\alpha_N}{\alpha_L}}|E_3,E_1)\Big] 
    \end{equation}
\end{itemize}

% For scenario 1, the probability of associating with NLoS UAV-BS, given there exists at least one \ac{NLoS} UAV-BS and no \ac{LoS} UAV-BS, is 1 i.e.,
% \begin{equation}
%   \mathbb{P}(d_L > d_N^{\frac{\alpha_N}{\alpha_L}})|E_3,E_2 =1  
% \end{equation}
The probability of associating to an \ac{NLoS} UAV-BS for scenario 1 is given as $A_N^{1}= (1-B_L)$. 

For scenario 2, we apply the \ac{CCDF} for $d_L$ from (\ref{cdf_z1}) and take an expectation over $d_N$.

We know the path loss in \ac{NLoS} links is much higher than that of \ac{LoS} links. Therefore, $\alpha_N$ $>>$ $\alpha_L$. 
The distance to the nearest \ac{NLoS} UAV-BS $d_N$ varies from $h_{min}$ to $\infty$. Therefore, the values of $d_L$ are always greater than $h_{max}$. Hence, the height of \ac{LoS} UAV-BS always vary from $h_{min}$ to $h_{max}$, since $\alpha_N$ is greater than $\alpha_L$.

Given that there exists at least one LoS UAV-BS and one NLoS UAV-BS, the distance to the LoS UAV-BS $d_L$ is greater than $d_N^{\frac{\alpha_N}{\alpha_L}}$if and only if there are no LoS UAV-BSs inside the 3-D region $\mathcal{B}$ with radius $\sqrt{d_N^{\frac{2\alpha_N}{\alpha_L}}-h_a^2}$ and height $h_a$.

Therefore, taking the CCDF of $d_L$ from (\ref{cdf_z1}).
\begin{equation}
\mathbb{P}(d_L >d_N^{\frac{\alpha_N}{\alpha_L}}|E_3,E_1)= 1- F_{d_L}(d_N^{\frac{\alpha_N}{\alpha_L}})
\end{equation}
\begin{multline}
   % 1- \Bigg(\frac{1-\exp{\Big(A  \times \int_{h_{min}} ^ {h_{max}} \mathcal{L}(d_N^{\frac{\alpha_N}{\alpha_L}},h_a) \hspace{0.1cm} \mathrm{d}h_a \Big)}}{B_L}\Bigg)\\
\frac{\bigg[\exp{\Big(A  \times \int_{h_{min}} ^ {h_{max}} \mathcal{L}(d_N^{\frac{\alpha_N}{\alpha_L}},h_a) \hspace{0.1cm} \mathrm{d}h_a \Big)}-(1-B_L)\bigg]}{B_L}
\end{multline}
where $A=\frac{-2\lambda_u \pi}{(h_{max}- h_{min})}$, $\mathcal{L}(z,h)$ is given in (\ref{bb1}).

Therefore, the probability of associating to an \ac{NLoS} UAV-BS for scenario 2 is given as
\begin{equation}
   A_N^{2}(d_N)= B_L
\frac{\bigg[\exp{\Big(A  \times \int_{h_{min}} ^ {h_{max}} \mathcal{L}(d_N^{\frac{\alpha_N}{\alpha_L}},h_a) \hspace{0.1cm} \mathrm{d}h_a \Big)}-(1-B_L)\bigg]}{B_L}
\end{equation}
Combining both scenarios, the probability of associating with an NLoS UAV-BS is given as
\begin{equation}
    A_N'(d_N)=A_N^{1} + A_N^{2}(d_N)
\end{equation}
\begin{multline}
    A_N'(d_N)= (1-B_L)+ \exp{\bigg(A  \times \int_{h_{min}} ^ {h_{max}} \mathcal{L}(d_N^{\frac{\alpha_N}{\alpha_L}},h_a) \hspace{0.1cm} \mathrm{d}h_a \bigg)} \\
    -1+B_L
\end{multline}
\begin{equation}
    =\exp{\bigg(A  \times \int_{h_{min}} ^ {h_{max}} \mathcal{L}(d_N^{\frac{\alpha_N}{\alpha_L}},h_a) \hspace{0.1cm} \mathrm{d}h_a \bigg)}
\end{equation}    

Further, taking an expectation over $d_N$,
% there will be two cases:
% \begin{enumerate}
%     \item $h_{min}^{\frac{\alpha_N}{\alpha_L}} \leq h_{max}$
%     \item $h_{min}^{\frac{\alpha_N}{\alpha_L}} > h_{max}$
% \end{enumerate}
% % For case 1, the expectation wrt to $d_N$ is given as
% \begin{multline}
%    A_N=\int_{h_{min}}^{h_{max}}\bigg[\exp{\bigg(A  \times \int_{h_{min}} ^ {r_1^{\frac{\alpha_N}{\alpha_L}}} \mathcal{L}(r_1^{\frac{\alpha_N}{\alpha_L}},h_a) \hspace{0.1cm} \mathrm{d}h_a \bigg)}\bigg]\\  f_{d_N}(r_1) \mathrm{d} r_1 +\\\int_{h_{max}}^{\infty} \bigg[\exp{\bigg(A  \times \int_{h_{min}} ^ {h_{max}} \mathcal{L}(r_2^{\frac{\alpha_N}{\alpha_L}},h_a) \hspace{0.1cm} \mathrm{d}h_a \bigg)}\bigg] \\ f_{d_N}(r_2) \mathrm{d} r_2 
% \end{multline}
% For case 2, the expectation wrt to $d_N$ is given as
\begin{multline}  A_N=\int_{h_{min}}^{h_{max}}\bigg[\exp{\bigg(A  \times \int_{h_{min}} ^ {h_{max}} \mathcal{L}(r_1^{\frac{\alpha_N}{\alpha_L}},h_a) \hspace{0.1cm} \mathrm{d}h_a \bigg)}\bigg] \\ f_{d_N}^{'}(r_1) \mathrm{d} r_1 +\int_{h_{max}}^{\infty} \bigg[\exp{\bigg(A  \times \int_{h_{min}} ^ {h_{max}} \mathcal{L}(r_2^{\frac{\alpha_N}{\alpha_L}},h_a) \hspace{0.1cm} \mathrm{d}h_a \bigg)}\bigg]  \\ f_{d_N}^{''}(r_2) \mathrm{d} r_2 
\end{multline}
where $f_{d_N}^{'}(r_1)$ and $f_{d_N}^{''}(r_2)$ are the derivative of CDF at (\ref{cdf_z2}).
% \begin{equation}
%   f_{d_N}^{'}(r_1)= \frac{\mathrm{d}}{\mathrm{d} r_1}  F_{d_N}(r_1) ; \quad  h_{min}\leq r_1 < h_{max}
% \end{equation}
% \begin{equation}
%   f_{d_N}^{''}(r_2)= \frac{\mathrm{d}}{\mathrm{d} r_2}  F_{d_N}(r_2) ; \quad  r_2 \geq h_{max}
% \end{equation}

% By using the transformation of the random variable, we perform the expectation with respect to $d_N$.

% To find $\mathbb{P}(d_L > d_N^{\frac{\alpha_N}{\alpha_L}})$, we transform $d_N^{\frac{\alpha_N}{\alpha_L}}$ into another random variable $Z'$. Then taking an expectation over $Z'$, the CCDF of $d_L$ is expressed as 
% $\mathbb{E}_{Z'}[\mathbb{P}(d_L>Z')]$.

% The pdf of the new random variable $Z'$ is found using the pdf of $d_N$. Applying the Jacobian of transformation, the pdf of $Z'$ is given as
% \begin{equation}
%     f_{Z'}(z')= \frac{\alpha_L}{\alpha_N} f_{d_N}\bigg(z'^{\frac{\alpha_L}{\alpha_N}}\bigg) z'^{\frac{\alpha_L}{\alpha_N}-1}
% \end{equation}
% Therefore, taking an expectation over $Z'$, the probability of associating to NLoS UAV-BS $A_N$ is given as
% \begin{multline}
% A_N=\int_{h_{min}^{\frac{\alpha_N}{\alpha_L}}}^{h_{max}^{\frac{\alpha_N}{\alpha_L}}}\bigg[\exp{\bigg(A  \times \int_{h_{min}} ^ {h_{max}} \mathcal{L}(z'_1,h_a) \hspace{0.1cm} \mathrm{d}h_a \bigg)}\bigg]\\  f_{Z^{'}}(z^{'}_{1}) \mathrm{d} z^{'}_{1} +\\\int_{h_{max}^{\frac{\alpha_N}{\alpha_L}}}^{\infty} \bigg[\exp{\bigg(A  \times \int_{h_{min}} ^ {h_{max}} \mathcal{L}(z'_2,h_a) \hspace{0.1cm} \mathrm{d}h_a \bigg)}\bigg] \\ f_{Z'}(z^{'}_{2}) \mathrm{d} z'_{2} 
% \end{multline}

\section{proof of Theorem 1}
The conditional success probability of associating to an \ac{LoS} UAV-BS is given as
\begin{gather}
    P_{SL}(\gamma)= 
    % \mathbb{P}(\mathrm{SINR} > \gamma | \Phi_U,\Phi_b)
    \mathbb{P}\bigg(\frac{P_u K g_l d_L^{-\alpha_L}}{\sigma_N + I_L'+I_N} > \gamma | \Phi_U,\Phi_b\bigg)
\end{gather}
where $I_L'$ and $I_N$ are the interfering strengths from the other \ac{LoS} and \ac{NLoS} UAV-BSs respectively, where $I_L'= \sum_{i:\textbf{X}_i \in  \Phi_L'} P_u K G_L^{'} d_L'^{-\alpha_L}$ and $I_N= \sum_{i:\textbf{X}_i \in  \Phi_N }P_u K g_N d_N^{-\alpha_N}$.

Conditioning on $\Phi_U$ and $\Phi_b$, we apply the \ac{CCDF} of the exponentially random variable $g_l$ and take the expectation over $g_l'$.
\begin{multline}
    P'_{SL}(\gamma) = \exp{\bigg(\frac{-\gamma \sigma_N}{P_u K d_L^{-\alpha_L}}\bigg)}  \prod_{i:\textbf{X}_i \in  \Phi_N}\bigg(\frac{1}{1+\frac{\gamma d^{-\alpha_N}_{N}}{d_L^{-\alpha_L}}} \bigg) \\ \prod_{i:\textbf{X}_i \in  \Phi_L'} \bigg(\frac{m}{m+\frac{\varepsilon \gamma d'^{-\alpha_L}_{L}}{d_L^{-\alpha_L}}}\bigg)^m
\end{multline}
The $b^{th}$ moment can be given as $M_{bL}{''}(d_L)= \mathbb{E}_{\Phi_b,\Phi_U}[(P'_{SL}(\gamma))^b]$. The product is taken over the points in the PPP $\Phi_U$. Therefore, we first simplify the expression by taking the expectation over $\Phi_b$ inside the product and then taking the expectation over $\Phi_U$.  The final expression, which involves expectations over both $\Phi_b$ and $\Phi_U$ is given as
\begin{multline}
    M_{bL}^{''}(d_L)= \exp{\bigg(\frac{-\gamma \sigma_N b }{P_u K d_L^{-\alpha_L}}\bigg)}\\ \mathbb{E}_{\Phi_N}\bigg[\prod_{i:\textbf{X}_i \in  \Phi_N}\mathbb{E}_{\Phi_b} \bigg(\frac{1}{1+\frac{\gamma d^{-\alpha_N}_{N}}{d_L^{-\alpha_L}}} \bigg)^{b} \bigg] \\ \mathbb{E}_{\Phi_L}\bigg[\prod_{i:\textbf{X}_i \in  \Phi_L}\mathbb{E}_{\Phi_b} \bigg(\frac{m}{m+\frac{\varepsilon \gamma d'^{-\alpha_L}_{L}}{d_L^{-\alpha_L}}}\bigg)^{mb} \bigg] 
\end{multline}
Taking the inner expectation wrt to $\Phi_b$ for the NLoS case,
\begin{multline}
L_S(\sqrt{x^2+t^2},t) \Big(\frac{1}{1+\frac{\gamma (\sqrt{x^2+t^2})^{-\alpha_L}}{z^{-\alpha_L}}}\Big)^b + \\ N_S(\sqrt{x^2+t^2},t) \Big(\frac{1}{1+\frac{\gamma (\sqrt{x^2+t^2})^{-\alpha_N}}{z^{-\alpha_L}}}\Big)^b.
\end{multline}
Taking the inner expectation wrt to $\Phi_b$ for the LoS case,
\begin{multline}
L_S(\sqrt{x^2+t^2},t) \Big(\frac{m}{m+\frac{\varepsilon \gamma (\sqrt{x^2+t^2})^{-\alpha_L}}{z^{-\alpha_L}}}\Big)^{mb} + \\ N_S(\sqrt{x^2+t^2},t) \Big(\frac{m}{m+\frac{\varepsilon \gamma (\sqrt{x^2+t^2})^{-\alpha_N}}{z^{-\alpha_L}}}\Big)^{mb},
\end{multline}
where $L_S(x,y)$ is the probability that the UAV-BS at $(x,y)$ is in LoS and $N_S(x,y)$ is the probability that the UAV-BS at $(x,y)$ is in NLoS. Next, we take an expectation over the height of the UAV-BSs for two cases of $d_L$.

% For $h_{min} \leq d_L < h_{max}$,
% \begin{multline}
%     \int_{\sqrt{z^2-x^2}}^{h_{max}}\bigg[ L_S(\sqrt{x^2+t^2},t) \Big(\frac{1}{1+\frac{\gamma (\sqrt{x^2+t^2})^{-\alpha_L}}{z^{-\alpha_L}}}\Big)^b + \\ N_S(\sqrt{x^2+t^2},t) \Big(\frac{1}{1+\frac{\gamma (\sqrt{x^2+t^2})^{-\alpha_N}}{z^{-\alpha_L}}}\Big)^b \bigg]\\ \frac{\mathrm{d} t}{(h_{max}- \sqrt{z^2-x^2})} 
% \end{multline}
% For $d_L \geq h_{max}$,
% \begin{multline}
% \int_{h_{min}}^{h_{max}} \bigg[ L_S(\sqrt{x^2+t^2},t) \Big(\frac{1}{1+\frac{\gamma (\sqrt{x^2+t^2})^{-\alpha_L}}{z^{-\alpha_L}}}\Big)^b + \\ N_S(\sqrt{x^2+t^2},t) \Big(\frac{1}{1+\frac{\gamma (\sqrt{x^2+t^2})^{-\alpha_N}}{z^{-\alpha_L}}}\Big)^b \bigg] \\\frac{\mathrm{d} t}{(h_{max}- h_{min})} 
% \end{multline}
Similarly, expectation over height is applied to the LoS case as well. 
Applying PGFL for the 2-D locations of the LoS \ac{UAV-BS} and NLoS UAV-BSs of the \ac{PPP} \cite{32}, the above equation can be written as (\ref{inn_1}) and (\ref{inn_2}). Taking an expectation over $d_L$ and multiplying the probability of associating to \ac{LoS} UAV-BS $A_L$ gives the $b^{th}$ moment of the \ac{CSP} of associating to \ac{LoS} \ac{UAV-BS} as given in (\ref{md1}).

\end{appendices}
\bibliography{references}
\bibliographystyle{IEEEtran}

\end{document}